\documentclass[journal]{IEEEtran}
\usepackage{amsmath,epsfig,epsf,times,amssymb,booktabs,subfigure}
\usepackage{setspace}
\usepackage{array}
\usepackage{url}
\usepackage{float}
\usepackage{color}
\usepackage{graphicx}
\usepackage{amsthm,cite}
\usepackage{enumerate}

\newtheorem{proposition}{Proposition}

\theoremstyle{remark}
\newtheorem{remark}{Remark}

\newcommand{\remarkend}{ \IEEEQEDopen}
\newtheorem{lemma}{Lemma}

\begin{document}

\title{Millimeter Wave Beam Alignment: Large Deviations Analysis and Design Insights}
\author{Chunshan Liu,~\IEEEmembership{Member,~IEEE}, Min Li,~\IEEEmembership{Member,~IEEE}, Stephen V. Hanly,~\IEEEmembership{Fellow,~IEEE},\\
Iain B. Collings,~\IEEEmembership{Fellow,~IEEE}, and Philip Whiting
\thanks{This research was supported in part by the Australian Research Council
under grant DP130101760, by the Macquarie University under the MQVCIF
Fellowship, and by the CSIRO Macquarie University Chair in Wireless
Communications. This Chair has been established with funding provided by
the Science and Industry Endowment Fund.
\par Chunshan Liu, Stephen V. Hanly, Iain B. Collings and Philip Whiting are with the Department of Engineering, Macquarie University, Sydney, NSW 2109, Australia (email: \{chunshan.liu, stephen.hanly, iain.collings, philip.whiting\}@mq.edu.au). Min Li is with the School of Electrical Engineering and Computing, The University of Newcastle, Callaghan, NSW 2308, Australia (email: min.li@newcastle.edu.au). (\textit{Corresponding author: Min Li}.)}}

\maketitle

\begin{abstract}
In millimeter wave cellular communication, fast and reliable beam alignment via beam training is crucial to harvest sufficient beamforming gain for the subsequent data transmission. In this paper, we establish fundamental limits in beam-alignment performance under both the exhaustive search and the hierarchical search that adopts multi-resolution beamforming codebooks, accounting for time-domain training overhead. Specifically, we derive lower and upper bounds on the probability of misalignment for an arbitrary level in the hierarchical search, based on a single-path channel model. Using the method of large deviations, we characterize the decay rate functions of both bounds and show that the bounds coincide as the training sequence length goes large. We go on to characterize the asymptotic misalignment probability of both the hierarchical and exhaustive search, and show that the latter asymptotically outperforms the former, subject to the same training overhead and codebook resolution. We show via numerical results that this relative performance behavior holds in the non-asymptotic regime. Moreover, the exhaustive search is shown to achieve significantly higher worst-case spectrum efficiency than the hierarchical search, when the pre-beamforming signal-to-noise ratio (SNR) is relatively low. This study hence implies that the exhaustive search is more effective for users situated further from base stations, as they tend to have low SNR.
\end{abstract}


\section{Introduction}\label{sec:introduction}
\par Millimeter wave (mmWave) communication is one of the important means to expand the system capacity for the fifth generation (5G) cellular networks, thanks to the abundant frequency bands in the range of $30$-$300$ GHz~\cite{boccardi2014five,andrews2014will,7010531}. Despite its great potential~\cite{rappaport2013millimeter,6834753}, many obstacles have to be overcome before mmWave cellular can be realized and deployed in practice. One of the key challenges is that mmWave induces much larger free-space loss due to its higher carrier frequency. Combatting the severe loss necessitates the use of large scale antenna arrays and beamforming at transceivers to accommodate directional transmission in mmWave systems~\cite{roh2014millimeter,alkhateeb2014mimo,heath2016overview}.

\par To achieve large beamforming gain, transmit and receive beams {at base station (BS) and user equipment (UE)} must be adaptively steered and aligned. Assuming perfect channel knowledge, recent studies~\cite{el2014spatially,alkhateeb2014mimo,heath2016overview,sohrabi2016hybrid} have developed a number of optimized hybrid analog and digital BS/UE beamforming solutions, subject to different hardware resource constraints. However, accurate estimation of the channel (i.e., all entries of the channel matrix) itself is a challenging task in mmWave communications, considering the large scale antenna arrays employed.

\par Another viable approach for beam alignment at mmWave is beam training, {which was first considered in the design of mmWave multi-Gigabit wireless local area network (WLAN)~\cite{cordeiro2010ieee,tsang2011coding} and wireless personal area network (WPAN)~\cite{wang2009beam} and has received significant attention in the design of mmWave cellular communications~\cite{hur2013millimeter,alkhateeb2014channel,Xiao2016,Gao2016,Kokshoorn2016Globecom,Kokshoorn2016}}. In this approach, BS and UE jointly examine BS/UE beamforming pairs from pre-designed codebooks that represent the beam search space to find strong multi-path components. This approach does not require explicit estimation of the channel coefficients, and is particularly useful to identify the dominant multi-path component in the sparse mmWave channels and align the transmission and reception along this path.

\par This paper focuses on the beam training approach and establishes fundamental limits in beam-alignment performance. In the literature, two different beam-training strategies have been considered. The conventional approach is for BS and UE to perform an exhaustive search by examining all beam pairs in the codebook and determine the best pair that maximizes a given performance metric (e.g., beamforming gain), which will be used for data transmission. The training overhead of this strategy is proportional to the size of the beam search space and thus can be prohibitive when narrow beams are employed with the aim of achieving high beamforming gain.

\par To reduce the search space (hence the training overhead), references~\cite{wang2009beam,hur2013millimeter,alkhateeb2014channel,Xiao2016} proposed a hierarchical beam search based on multi-level codebook designs (see Fig.~\ref{fig:system}(b) for an example). In such designs, a lower-level codebook consists of wider beams that jointly cover the intended angular interval. A higher-level codebook consists of narrower beams covering the same angular interval, and a subset of them cover the interval of a wider beam in the lower-level. The hierarchical beam search, in the spirit of bisection search, first finds the best wide-beam pair in a lower-level codebook, and then iteratively refines the search using the next-level codebook within the beam subspace prescribed by the wide-beam pair found.

\par {Practical codebooks for the hierarchical search can be realized with multiple RF chains~\cite{alkhateeb2014channel} or a single RF chain~\cite{Xiao2016}. Whereas having multiple RF chains might provide more flexibility in synthesizing the desired beams, in~\cite{Xiao2016}, it was shown that a cost-effective beam synthesis technique using a single RF chain was capable of producing flatter beams with better performance than alternatives realized with more RF chains.} With synthesized beams, it has been shown that hierarchical search can approach the performance of exhaustive search but with fewer training iterations~\cite{alkhateeb2014channel,Xiao2016}.

\par However, a common premise underlying these works is that each training stage embraces relatively high signal-to-noise ratio (SNR) so that each stage tends to achieve the same beam-alignment performance. As indicated by~\cite{alkhateeb2014channel,Xiao2016}, this might be possible if there is sufficiently high transmit power or a sufficiently long training sequence is employed at each stage. But in general, given finite training resources (power and time), an early stage with weak beamforming gains is likely to experience relatively low SNR, even for typical small cell coverage range on the order of $100$ meters (see the measurement results reported in~\cite{6834753}). There may thus be a higher chance of failing to find the best beam pair at an early stage, leading to subsequent misalignment at higher levels.

\par {To alleviate this potential issue in the hierarchial search, some remedies were proposed in~\cite{alkhateeb2014channel} and in more recent studies~\cite{Kokshoorn2016Globecom,Kokshoorn2016}. In particular, reference~\cite{alkhateeb2014channel} proposed to allocate power inversely proportional to beamforming gains at different stages with the aim of equalizing the performance at each stage. Whilst this is a valid approach in theory, it does induce a large peak-to-average power ratio, which might not be tolerable in practical system design. References~\cite{Kokshoorn2016Globecom,Kokshoorn2016} proposed an adaptive beam-alignment and channel estimation approach that allows re-training (i.e., additional measurements) to be taken within the most likely subrange pair identified through initial training at each stage until a {certain probability threshold has been met or the maximum number of measurements has been reached}. This approach hence would involve dynamic allocation of the number of measurements between different stages and would require additional one or multiple rounds of feedback for each stage, compared to the hierarchical search of~\cite{alkhateeb2014channel}. Performance upper bounds were derived in these works, but no analytical comparison was established between the hierarchical search and the exhaustive search, subject to the same finite amount of training resource.}

\par In this paper, we establish fundamental limits and comparison on the beam-alignment performance under both the exhaustive search and the hierarchical search, accounting for the same amount of training power and pilot symbols. {In order to minimize the transmit power dynamic and overcome the poor efficiency of power-amplifiers at mmWave bands~\cite{huang2011millimeter,yong201160ghz}, we assume each training symbol is transmitted at constant peak power. We also assume that both BS and UE employ a single RF chain and thus the beam codewords for training purpose are realized using superior analog beamforming proposed by~\cite{Xiao2016}.}

\par We note that if the SNR is sufficiently high where a short pilot sequence is enough to ensure satisfactory low misalignment probability, the hierarchical search can be effective and has the advantage to reduce the search space. However, for lower SNRs, both the exhaustive search and the hierarchical search require longer pilot sequence to ensure a low probability of misalignment. In such circumstances, given a fixed and finite number of training pilot symbols, the dilemma is that the hierarchical search allocates more pilot symbols to each beam pair, due to the smaller number of beam pairs to be examined, but it starts with wide beams with relatively low beam forming gain in the first stage, which may fail and become the bottleneck on the overall beam-alignment performance; on the other hand, the exhaustive search needs to examine more beam pairs, but it goes directly to narrow beams with relatively high beam forming gain, which may offset the disadvantage in training resource per beam pair and achieve a better beam-alignment performance. It is thus unclear how the training resource considered affects the performance of these two search strategies and how they compare with each other subject to the same amount of training resource.

\par The main contributions of the paper are summarized as follows:
\begin{itemize}
  \item  We derive both lower and upper bounds on the probability of misalignment at an arbitrary level in the hierarchical search, focusing on a single-path channel model. The bounds are characterized using the training sequence length and other relevant system parameters. Using the method of large deviations~\cite{dembo2009large}, we further characterize the decay rate functions of these bounds and show that the upper and lower bounds coincide as the training sequence length goes large. As an important byproduct of this analysis, we identify that a desirable codeword for beam training purposes should have uniform gain in its intended coverage interval and zero leakage outside the interval (which we refer to as an ideal beam pattern).
  \item Building on the results above, we characterize the asymptotic misalignment probability of both hierarchical and exhaustive search. Under the same targeted beam search resolution, we show that hierarchical search is outperformed by exhaustive search as the training sequence length goes large, provided that the same training overhead is incurred. Moreover, numerical results demonstrate that this relative performance behavior is valid in the non-asymptotic regime when the pre-beamforming SNR is relatively low.
  \item We extend the study to more general setups involving multi-path channels and imperfect beamforming patterns. It is numerically shown that with practical synthesized training codebooks, the exhaustive search still significantly outperforms the hierarchical search in terms of either misalignment probability or achievable spectral efficiency (for data transmission after beam training) when the pre-beamforming SNR is relatively low, under the same amount of training overhead. The performance difference diminishes in the high SNR regime, as both schemes can provide high reliability of beam-alignment.
\end{itemize}
\par \textit{Notation}: Boldface uppercase letters and boldface lowercase letters are used to denote matrices
and vectors, respectively, e.g., ${\bf A}$ is a matrix and ${\bf a}$ is a vector. Notations $(\cdot)^T$ and $(\cdot)^\dag$ denote
transpose and conjugate transpose, respectively. Notation $\|a\|_2$ stands for the $l_2$ norm of vector ${\bf a}$, and notation ${\mathbb E}(\cdot)$ denotes the expectation operation. ${\cal CN}(0,\sigma^2)$ denotes a zero-mean circularly-symmetric complex Gaussian distribution with variance $\sigma^2$. Finally, $\left|{\cal A}\right|$ denotes the cardinality of set ${\cal A}$ and ${\mathbb C}$ is the set of all complex numbers.

\section{System Description and Problem Formulation}\label{sec:model}
We consider a mmWave downlink communication, in which the BS, equipped with $N_T$ antennas, wishes to communicate with a UE that is equipped with $N_R$ antennas, as depicted in Fig.~\ref{fig:system}(a). Both BS and UE are provided only a single RF chain, and thus analog BS/UE beamforming is adopted, similar to~\cite{wang2009beam,Xiao2016}. {Depending on the beam synthesis techniques, each antenna at BS/UE is allowed to be turned off (or deactivated) if needed.} In order to align BS/UE beams, a hierarchical training approach that employs multi-resolution beam codebooks is considered. In the beam alignment process, we assume that the BS and the UE are synchronized (we refer to \cite{barati2014dreictional,shokri2015millimeter,LiummWave2016} for possible mmWave synchronization/initial access solutions), and the UE can feedback the transmit beamforming indices to the BS when necessary.

In the following,  we first define the beam training codebooks and then formalize the beam search operations and relevant signal models. {Table~\ref{table:notations} lists the key notations related to the model description and problem formulation below.}

\begin{figure}[t]
\centering
\includegraphics[width=0.45\textwidth]{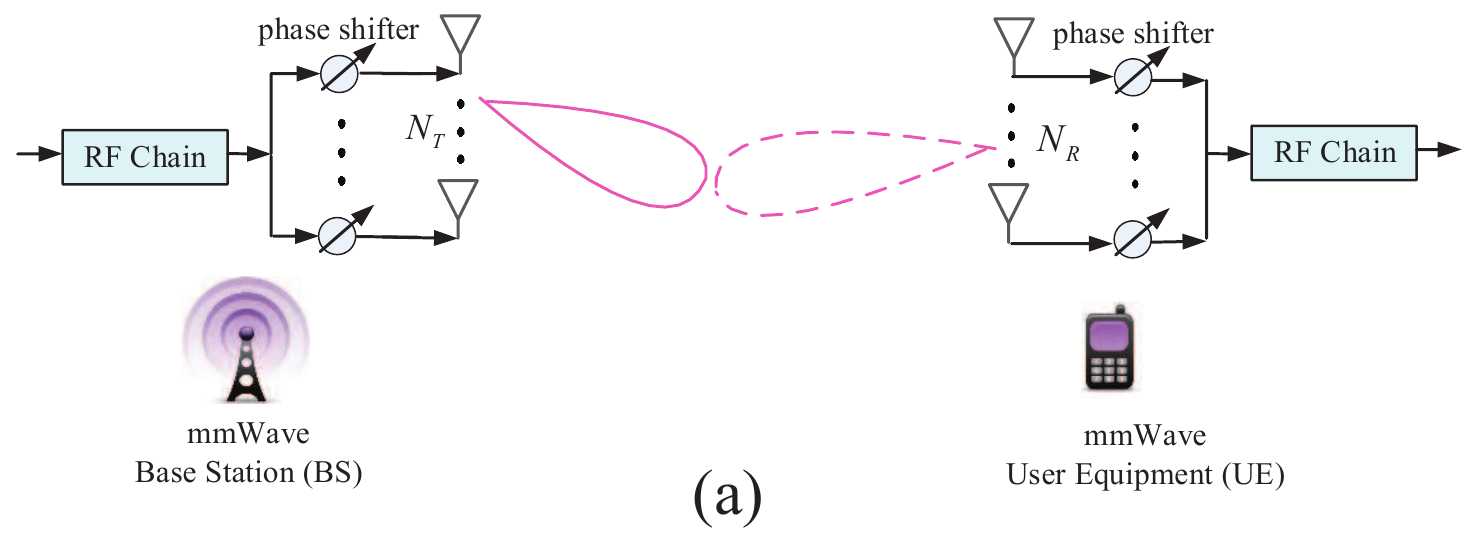}
\includegraphics[width=0.4\textwidth]{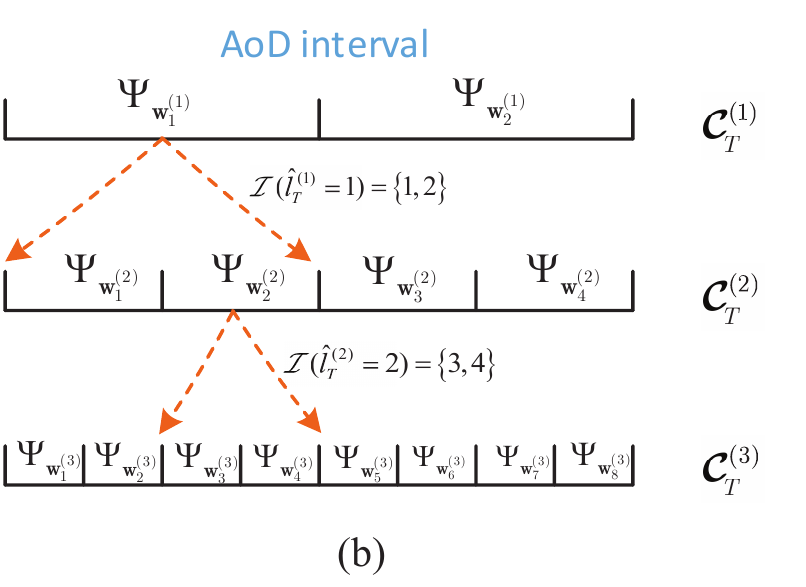}
\caption{An illustration of (a): the mmWave communication system considered, and (b): the hierarchical codebooks at BS with $K=3$ and $L_T^{(k)} = 2^k$ ($k=1,2,3$).}
\label{fig:system}
\end{figure}

\par Let $\Psi$ and $\Phi$ be the entire {Angle of Departure (AoD) and Angle of Arrival (AoA)} search interval, respectively. Note that $\Psi$ and $\Phi$ could cover the entire angular space or only a portion (e.g., a $60^\circ$ sector between $[-30^\circ, +30^\circ]$).  Let $\mathbf{w}_{l_T}^{(k)} \in \mathbb{C}^{1\times N_T}$ and $\mathbf{f}_{l_R}^{(k)} \in \mathbb{C}^{1\times N_R}$ denote an arbitrary BS beamformer and UE combiner at the $k$th level, with coverage intervals
$\Psi_{\mathbf{w}_{l_T}^{(k)}}$ and $\Phi_{\mathbf{f}_{l_R}^{(k)}}$, respectively. In what follows, we focus on elaborating the codebook structure at the BS, as the UE codebooks can be constructed in a similar manner.

\begin{table}[t]
  \centering
  \caption{{A List of Key Notations Related to the Model Description and Problem Formulation}}
    \begin{tabular}{l l}
    \hline
    $_l$      &index of a transmit/receive beam pair \\
    $^{(k)}$  &index of a search level in the hierarchical search\\
    $L_T^{(k)}$ &transmit codebook size at level $k$ \\
    $L_R^{(k)}$ &receive codebook size at level $k$ \\
    $L^{(k)}$ &number of beam pairs examined at level $k$ \\
    $L$ & number of beam pairs examined in $K$-level search \\
    $L_{ex}$ &number of beam pairs examined in exhaustive search \\
    $l_{opt}^{(k)}$ &index of the optimal beam pair at level $k$ \\
    $\hat{l}_{opt}^{(k)}$ &index of the estimated best beam pair at level $k$ \\
    \hline
    \end{tabular}%
  \label{table:notations}%
\end{table}

\par Let ${\cal C}_{T}^{(k)} = \{{\mathbf{w}_{l_T}^{(k)}}, l_T =1,\cdots, L_T^{(k)}\}$ denote the entire codebook of level $k$ at BS, whose size equals $L_T^{(k)}$. A $K$-level hierarchical search is performed based on a sequence of codebooks ${\cal C}_{T}^{(1)},\cdots,{\cal C}_{T}^{(K)}$ that are arranged in a nested structure~\cite{Xiao2016,alkhateeb2014channel}. An example of such a codebook design with $K=3$ and $L_T^{(k)} = 2^k$ ($k=1,2,3$) is depicted  in  Fig.~\ref{fig:system} (b), where the coverage intervals of the beamformers are illustrated by line segments. In general, the codewords at each level~$k$ jointly cover the entire search interval and the number of codewords tends to increase as $k$ grows with the aim of providing higher spatial resolution. We thus have $\bigcup_{l_T=1}^{L_T^{(k)}}\Psi_{\mathbf{w}^{(k)}_{l_T}} = \Psi,~k=1,\cdots,K$, and $L_T^{(1)} \le L_T^{(2)} \le \cdots \le L_T^{(K)}$. Moreover, the codebooks are nested in the sense that the coverage of an arbitrary codeword ${\mathbf{w}^{(k)}_{{\hat l}_T}}$ at lower-level $k$ is covered by the union of a number of codewords at the higher-level $k+1$. Denote the collection of these codeword indices at level $k+1$ as $\mathcal{I}({\hat l}_T^{(k)})$. We then have that $\bigcup_{l_T \in \mathcal{I}({\hat l}_T^{(k)})}\Psi_{\mathbf{w}^{(k+1)}_{l_T}} = \Psi_{\mathbf{w}^{(k)}_{{\hat l}_T}},~k=1,\cdots,K$. It is further assumed that $\mathcal{I}({\hat l}_T^{(k)}) \cap \mathcal{I}({\tilde l}_T^{(k)}) = \emptyset$, for any ${\hat l}_T^{(k)} \ne {\tilde l}_T^{(k)}$.

\par With a series of such codebooks at BS and UE, the hierarchical beam search starts with the first level codebook, examines all possible combinations of beamformer/combiner, and determines the best pair that produces the maximum output signal. Assume that this best pair is indexed by $({\hat l}_T^{(1)},{\hat l}_R^{(1)})$. The UE then feedbacks beamformer index ${\hat l}_T^{(1)}$ to the BS. In the next stage, the search will only probe into a subset of beamformers/combiners in the codebooks (i.e., $\{{\mathbf{w}^{(2)}_{l_T}} \times {\mathbf{f}^{(2)}_{l_R}}: l_T \in \mathcal{I}({\hat l}_T^{(1)}), l_R \in \mathcal{I}({\hat l}_R^{(1)})\}$), and determine the best pair $({\hat l}_T^{(2)},{\hat l}_R^{(2)})$ that produces the maximum output signal, and then the UE feedbacks index ${\hat l}_T^{(2)}$ to the BS. Given $({\hat l}_T^{(2)},{\hat l}_R^{(2)})$, the search turns to evaluate a corresponding subset of beamformers/combiners in the next-level codebooks. This search continues recursively until it reaches the last level of the codebooks.

\par Let $L^{(k)} = \left|\mathcal{I}({\hat l}_T^{(k-1)})\right|\left|\mathcal{I}({\hat l}_R^{(k-1)})\right|$ be the total number of beam pairs examined at the $k$th level (with $L^{(1)}=L_T^{(1)}L_R^{(1)}$). The training overhead of a $K$-level hierarchical search thus amounts to $L = \sum_{k=1}^{K}L^{(k)}$. For comparison, consider an exhaustive search that targets the same beam resolution of the hierarchical search. The training overhead would be $L_{ex} = L_T^{(K)}L_R^{(K)}$, which can be significantly larger than $L$, as $L_{ex} \gg L^{(k)}$, $\forall k=1,\ldots,K$. This search space reduction is the potential advantage of the hierarchical search.

\par In order to characterize the performance of the hierarchical search, we now formalize signal models relevant to the beam training and search operations. We adopt a frequency-flat and block fading channel model, and assume that the channel does not change during the beam-alignment process. For ease of exposition, we relabel the $L^{(k)}$ pairs of beamformer/combiner $\{{\mathbf{w}^{(k)}_{l_T}} \times {\mathbf{f}^{(k)}_{l_R}}: l_T \in \mathcal{I}({\hat l}_T^{(k-1)}), l_R \in \mathcal{I}({\hat l}_R^{(k-1)})\}$ to be evaluated at level $k$ as $\{{\mathbf{w}^{(k)}_{l}} \times {\mathbf{f}^{(k)}_{l}}: l = 1, \cdots,L^{(k)}\}$.
\par Let $\mathbf{s} \in \mathbb{C}^{1\times N}$ be a training pilot sequence with $N$ symbols {allocated to each beam pair. Then the total number of pilot symbols, $N_{tot}$, amounts to the product of the number of beam pairs examined in the search process and the pilot sequence length $N$ for each beam pair.} Considering any arbitrary level $k$ in the hierarchical search, the received signal associated with the $l$th beamformer/combiner pair can be represented as:
\begin{align}
\mathbf{y}^{(k)}_{l} &= \mathbf{f}^{(k)}_{l}\mathbf{H}\mathbf{w}^{{(k)}\dagger}_{l}\mathbf{s} + \mathbf{f}^{(k)}_{l}\mathbf{Z}^{(k)}_{l} \nonumber\\
&= h^{(k)}_l\mathbf{s} + \mathbf{z}^{(k)}_{l}\in \mathbb{C}^{1\times N},~~~l = 1,\ldots,L^{(k)}.
\end{align}
where $\mathbf{H}\in \mathbb{C}^{N_R\times N_T}$ is the mmWave channel matrix, $\mathbf{Z}^{(k)}_{l} \in \mathbb{C}^{N_R \times N}$ is the noise matrix (before receiver combining) with i.i.d. components $\sim {\cal CN}(0,\sigma^2)$, and we have defined
\begin{align}
h^{(k)}_{l} \triangleq \mathbf{f}^{(k)}_{l}\mathbf{H}\mathbf{w}^{(k)\dagger}_{l}, \label{equ:effective:channel}
\end{align}
as the effective channel after any fixed transmit beamforming and receive combining. Both BS beamformer and UE combiner are {realized using a single RF chain, and their entries are either of constant modulus or zero (i.e., allowing deactivation of some antennas to form wide beams). Moreover, both beamformer and combiner are assumed to be of unit-norm, namely, $\|\mathbf{w}^{(k)}_{l}\|_2^2 = \|\mathbf{f}^{(k)}_{l}\|_2^2=1$.} Therefore, $\mathbf{z}^{(k)}_{l} \in \mathbb{C}^{1\times N}$ with i.i.d. components $\sim {\cal CN}(0,\sigma^2)$. We further assume that the pilot sequence is transmitted at {constant} peak power $P_T$, {which aims to minimize the transmit power dynamic and overcome the poor power-efficiency of power amplifiers at mmWave frequencies~\cite{huang2011millimeter,yong201160ghz}}. Thus $\|\mathbf{s}\|_2^2=NP_T$.

\par The received signal is then match-filtered with training sequence $\mathbf{s}$, and the best beamformer/combiner pair at arbitrary level $k$ is selected as the one that gives rise to the strongest match-filtering output~\cite{Xiao2016}:
\begin{equation}\label{Eq:est_1}
\hat{l}_{opt}^{(k)} = \arg\max_{l=1,\ldots,L^{(k)}}\left|\mathbf{y}^{(k)}_l\mathbf{s}^{\dagger}\right|.
\end{equation}
\par In the absence of noise, it should be highly reliable to identify the best BS and UE beam pair that is aligned with the dominant path at an arbitrary codebook level. However, the output signal can be swayed by the noise if it is present. In the next section, we shall analyze the property of \eqref{Eq:est_1} and characterize the beam search performance at an arbitrary level~$k$ as well as the overall beam-alignment performance.

\par {Before we delve into the analysis, we remark that the match-filtering output~\eqref{Eq:est_1} essentially captures the energy received by UE, which depends on the pilot sequence length, the transmit power, as well as the effective channel gain. A stronger match-filtering output is thus expected by using a larger transmit power, or a longer pilot sequence length, or both. The analysis in this work focuses on the impact of pilot sequence length that tends to be large. This is motivated by the fact that the coherence bandwidth of mmWave band can be on the order of several~ten~MHz or more~\cite{rappaport2015wideband}, and thus the number of symbols accommodated within a coherent time interval can be very large. The asymptotic analysis in the current work is hence closely relevant to this practical scenario and is useful to provide guideline on practical system designs. Similar analysis in theory continues to apply when the transmit power goes large, but this would provide little practical implication, as peak-power constraint is a critical factor in mmWave communications as discussed earlier in Section~\ref{sec:introduction}.}

\section{Performance Analysis Under Rank-One (Single-Path) Channel Model}
In this section, we present fundamental limits on the performance of mmWave beam-alignment and compare the hierarchical search with the exhaustive search. Throughout the analysis, we consider a deterministic channel and focus on the rank-one channel model (when the channel has only single-path) for tractability. Later, we shall show by simulations that the key insights generated from the analysis with single-path channels continue to apply to scenarios in which {there are multi-paths with a dominant path, i.e., typical mmWave Rician channels}, see Section~\ref{sec:simulation}.

\par For the case with a single-path between BS and UE, the channel matrix $\mathbf{H}$ can be represented as:
\begin{equation}\label{eq:ch_singlepath}
\mathbf{H} = \alpha \mathbf{u}^{\dagger}(\phi)\mathbf{v}(\psi),
\end{equation}
where $|\alpha|^2$ is the path gain and $\mathbf{u}(\phi)\in \mathbb{C}^{1\times N_R}$ and $\mathbf{v}(\psi)\in \mathbb{C}^{1\times N_T}$ are the steering vectors corresponding to AoA $\phi$ and AoD $\psi$, respectively.
For instance, if a uniform linear array with half wave-length antenna spacing is equipped at the BS, the steering vector $\mathbf{v}(\psi)$ can be represented as:
\begin{equation}
\mathbf{v}(\psi) = [1,e^{j\pi\sin(\psi)},\ldots,e^{j\pi(N_T-1)\sin(\psi)}].
\end{equation}
Therefore, from~\eqref{equ:effective:channel} and~\eqref{eq:ch_singlepath}, the effective channel gain under any fixed transmit and receive beamforming can be represented as:
\begin{align}
g^{(k)}_{l}\triangleq |h^{(k)}_l|^2 &=|\alpha (\mathbf{f}^{(k)}_{l}\mathbf{u}^{\dagger}(\phi))(\mathbf{v}(\psi)\mathbf{w}_{l}^{(k)\dagger})|^2\\
&= |\alpha|^2W^{(k)}_{l}(\psi)F^{(k)}_{l}(\phi)\\
&= |\alpha|^2G_l^{(k)}(\psi,\phi),
\end{align}
where $W^{(k)}_{l}(\psi) \triangleq |\mathbf{v}(\psi)\mathbf{w}^{(k)\dagger}_{l}|^2$ is the transmit beamforming gain of $\mathbf{w}^{(k)}_l$ at AoD $\psi$, $F^{(k)}_{l}(\phi) \triangleq |\mathbf{u}(\phi)\mathbf{f}^{(k)\dagger}_{l}|^2$ is the receive beamforming gain of $\mathbf{f}^{(k)}_{l}$ at AoA $\phi$, and $G_l^{(k)}(\psi,\phi) \triangleq W^{(k)}_{l}(\psi)F^{(k)}_{l}(\phi)$ is the corresponding combined transmit and receive beamforming gain.

By the properties of the hierarchical codebooks, we have
\begin{equation}
g^{(k)}_{l}>g^{(k)}_{j}, ~\forall j\neq l,~\text{if}~\psi\in\Psi_{\mathbf{w}^{(k)}_{l}},\phi\in\Phi_{\mathbf{f}^{(k)}_{l}}.
\end{equation}
Namely, the highest possible effective channel gain at level $k$ is achieved by beamformer $\mathbf{w}^{(k)}_{l}$ and combiner $\mathbf{f}^{(k)}_{l}$ if the AoD ($\psi$) and AoA ($\phi$) of the propagation path is within the coverage of beamformer $\mathbf{w}^{(k)}_{l}$ and combiner $\mathbf{f}^{(k)}_{l}$.

\par Let $l^{(k)}_{opt} \triangleq \arg\max_{l=1,\ldots,L^{(k)}}g_{l}^{(k)}$ be the index of the optimal beam pair at level $k$, which achieves the highest effective channel gain at this level. It is then clear that a misalignment event occurs whenever the estimated best beam-pair index $\hat{l}_{opt}^{(k)} \ne l^{(k)}_{opt}$, for any $k=1,\cdots, K$.

\subsection{Probability of Misalignment at the $k$th Level}\label{sub_sec_single}
In this subsection, we focus on an arbitrary level $k$. To facilitate analysis, we introduce a normalized version of statistic~\eqref{Eq:est_1} that is defined as
\begin{equation}\label{Eq:est_2}
T^{(k)}_l\triangleq \frac{2\left|\mathbf{y}^{(k)}_l\mathbf{s}^{\dagger}\right|^2}{\sigma^2 \|\mathbf{s}\|_2^2}.
\end{equation}
It is straightforward to see that
\begin{align}
\hat{l}_{opt}^{(k)} = \arg\max_{l=1,\ldots,L^{(k)}}\left|\mathbf{y}^{(k)}_l\mathbf{s}^{\dagger}\right| = \arg\max_{l=1,\ldots,L^{(k)}} T^{(k)}_l.
\end{align}

\par The probability of misalignment at level $k$ can then be represented as
\begin{align}\label{Eq:miss_align_1}
p^{(k)}_{miss} &= Pr\{\hat{l}_{opt}^{(k)}\neq l^{(k)}_{opt}\} =  Pr\left\{T_{l^{(k)}_{opt}}^{(k)} < T_{\hat{l}_{opt}^{(k)}}^{(k)}\right\}\nonumber\\
& = Pr\left\{\bigcup_{l=1,l\neq l^{(k)}_{opt}}^{L^{(k)}} T_{l^{(k)}_{opt}}^{(k)}< T_{l}^{(k)}\right\}.
\end{align}
Using the union bound, it can be seen that $p^{(k)}_{miss}$ is upper bounded by:
\begin{equation}\label{Upper_bound}
p^{(k)}_{miss} \leq \sum_{l=1,l\neq l^{(k)}_{opt}}^{L^{(k)}} Pr\left\{T_{l^{(k)}_{opt}}^{(k)}< T_{l}^{(k)}\right\}\triangleq p^{(k)}_{up}.
\end{equation}
\par It is noted that $T^{(k)}_l$'s are independent as each of them is measured at a different time under different transmit and receive beamforming combinations. Moreover, each $T^{(k)}_l$ admits a non-central chi-square distribution with $2$ degrees of freedom~(DoFs) and a non-centrality parameter
\begin{align}
\lambda^{(k)}_{l} = \frac{2|h^{(k)}_{l}|^2\|\mathbf{s}\|^2_2}{\sigma^2}=\frac{2NP_Tg_{l}^{(k)}}{\sigma^2},
\end{align}
i.e., we have
\begin{equation}
T^{(k)}_{l} = \frac{2\left|\mathbf{y}^{(k)}_{l}\mathbf{s}^{\dagger}\right|^2}{\sigma^2\left\|\mathbf{s}\right\|_2^2} \sim \chi^2_2(\lambda^{(k)}_{l}).
\end{equation}
These facts lead to the conclusion that the ratio of $T_{l^{(k)}_{opt}}^{(k)}$ to ${T_{l}^{(k)}}$ (i.e., {$T_{l^{(k)}_{opt}}^{(k)}/{T_{l}^{(k)}}$)} admits a doubly non-central \emph{F}-distribution (denoted by $F(n_1,n_2,\eta_1,\eta_2)$), with DoFs $n_1 =n_2 =2$ and non-centrality parameters $\eta_1 = \lambda^{(k)}_{l^{(k)}_{opt}}$ and $\eta_2 = \lambda^{(k)}_{l}$. This allows us to represent the pair-wise probability $Pr\left\{T_{l^{(k)}_{opt}}^{(k)}< T_{l}^{(k)}\right\}$ using the cumulative distribution function (CDF) of the doubly non-central \emph{F}-distribution:
\begin{equation}\label{Piece_wise_up}
Pr\left\{T_{l^{(k)}_{opt}}^{(k)}< T_{l}^{(k)}\right\} = F_{(2,2)}\left(1|\lambda_{l^{(k)}_{opt}}^{(k)},\lambda_{l}^{(k)}\right),
\end{equation}
where $F_{(n_1,n_2)}(x|\eta_1,\eta_2)$ is the CDF of $F(n_1,n_2,\eta_1,\eta_2)$. The upper bound $p^{(k)}_{up}$ in \eqref{Upper_bound} can then be represented as:
\begin{equation}\label{upper_bound_f}
p^{(k)}_{up} = \sum_{l=1,l\neq l^{(k)}_{opt}}^{L^{(k)}} F_{(2,2)}\left(1|\lambda_{l^{(k)}_{opt}}^{(k)},\lambda_{l}^{(k)}\right).
\end{equation}
\par On the other hand, applying the Bonferroni inequalities~\cite{galambos1977bonferroni} to \eqref{Eq:miss_align_1}, it can be shown that $p^{(k)}_{miss}$ is lower bounded as:
\begin{align}
p^{(k)}_{miss} & \geq p^{(k)}_{up} - \sum_{\scriptstyle 1 \leq i<j \leq L^{(k)}\atop \scriptstyle i,j\neq l^{(k)}_{opt} }Pr\left\{T_{l^{(k)}_{opt}}^{(k)}< T_{i}^{(k)}, T_{l^{(k)}_{opt}}^{(k)}< T_{j}^{(k)}\right\} \\
& \geq p^{(k)}_{up} - \sum_{\scriptstyle 1\leq i<j \leq L^{(k)}\atop \scriptstyle i,j\neq l^{(k)}_{opt} } Pr\left\{T_{l^{(k)}_{opt}}^{(k)}< \frac{T_{i}^{(k)}+T_{j}^{(k)}}{2}\right\}  \label{lower_bound} \\
&\triangleq p^{(k)}_{low}, \nonumber
\end{align}
where \eqref{lower_bound} holds as $Pr\left\{T_{l^{(k)}_{opt}}^{(k)}< T_{i}^{(k)}, T_{l^{(k)}_{opt}}^{(k)}< T_{j}^{(k)}\right\} \leq Pr\left\{T_{l^{(k)}_{opt}}^{(k)}< \frac{T_{i}^{(k)}+T_{j}^{(k)}}{2}\right\}$. Since $T_{i}^{(k)}$ and $T_{j}^{(k)}$ are independent chi-square random variables, random variable $T_{i}^{(k)}+T_{j}^{(k)}$ also admits a non-central chi-square distribution, with $4$ DoFs and a non-centrality parameter $\lambda^{(k)}_{i}+\lambda^{(k)}_{j}$. Following a similar treatment as we did in \eqref{Piece_wise_up}, we can further represent the lower-bound as:
\begin{equation}\label{lower_bound_f}
p^{(k)}_{low} = p^{(k)}_{up} - \sum_{\scriptstyle 1\leq i<j \leq L^{(k)}\atop \scriptstyle i,j\neq l^{(k)}_{opt}}F_{(2,4)}\left(1|\lambda_{l^{(k)}_{opt}}^{(k)},\lambda_{i}^{(k)}+\lambda_{j}^{(k)}\right).
\end{equation}

The upper bound in \eqref{upper_bound_f} and the lower bound in \eqref{lower_bound_f} are characterized using the pair-wise probabilities $F_{(2,2)}\left(1|\lambda_{l^{(k)}_{opt}}^{(k)},\lambda_{l}^{(k)}\right)$ and $F_{(2,4)}\left(1|\lambda_{l^{(k)}_{opt}}^{(k)},\lambda_{i}^{(k)}+\lambda_{j}^{(k)}\right)$, which depend on the length of the pilot sequence ($N$). By the large deviation principle (LDP)~\cite{dembo2009large}, it can be shown that $F_{(2,2)}\left(1|\lambda_{l^{(k)}_{opt}}^{(k)},\lambda_{l}^{(k)}\right)$ and $F_{(2,4)}\left(1|\lambda_{l^{(k)}_{opt}}^{(k)},\lambda_{i}^{(k)}+\lambda_{j}^{(k)}\right) $ both decay exponentially as $N \rightarrow \infty$. This implies that as $N$ increases, the probability of misalignment decreases, which is consistent with intuition. Moreover, it can be shown that as $N \rightarrow \infty$, the upper bound and the lower bound coincide, which provides a useful means to analyze the behavior of $p^{(k)}_{miss}$. We present these results in Lemma~\ref{Lemma_LDP} and Proposition~\ref{Proposition_1}.

\begin{lemma}\label{Lemma_LDP}
Define $\xi^{(k)}_{l}\triangleq \frac{\lambda^{(k)}_{l}}{N} = \frac{2P_Tg^{(k)}_{l}}{\sigma^2}$. Then the pair-wise probabilities $F_{(2,2)}\left(1|\lambda_{l^{(k)}_{opt}}^{(k)},\lambda_{l}^{(k)}\right)$ and $F_{(2,4)}\left(1|\lambda_{l^{(k)}_{opt}}^{(k)},\lambda_{i}^{(k)}+\lambda_{j}^{(k)}\right) $ satisfy:
\begin{equation}\label{eq:lemma1_1}
\lim_{N \uparrow \infty}\frac{1}{N}\log F_{(2,2)}\left(1|\lambda_{l^{(k)}_{opt}}^{(k)},\lambda_{l}^{(k)}\right) = -I_1\left(\xi_{l^{(k)}_{opt}}^{(k)},\xi_{l}^{(k)}\right),
\end{equation}
and
\begin{align}\label{eq:lemma1_2}
\lim_{N \uparrow \infty}\frac{1}{N}\log &F_{(2,4)}\left(1|\lambda_{l^{(k)}_{opt}}^{(k)},\lambda_{i}^{(k)}+\lambda_{j}^{(k)}\right) \nonumber \\
&~~~~~~~~~~~~~~~=-I_2\left(\xi_{l^{(k)}_{opt}}^{(k)},\xi_{i}^{(k)},\xi_{j}^{(k)}\right),
\end{align}
respectively, where
\begin{equation}\label{rate_1}
I_1(\xi_{l^{(k)}_{opt}}^{(k)},\xi_{l}^{(k)}) = \frac{\left(\sqrt{\xi^{(k)}_{l^{(k)}_{opt}}}-\sqrt{\xi^{(k)}_{l}}\right)^2}{4}
 \end{equation}
 and
 \begin{equation}\label{rate_2}
 I_2\left(\xi_{l^{(k)}_{opt}}^{(k)},\xi_{i}^{(k)},\xi_{j}^{(k)}\right) = \frac{\left(\sqrt{2\xi^{(k)}_{l^{(k)}_{opt}}}-\sqrt{\xi^{(k)}_{i}+\xi^{(k)}_{j}}\right)^2}{6}.
 \end{equation}
\end{lemma}
\begin{proof}
See Appendix~\ref{proof_LDP}.
\end{proof}

\begin{proposition}\label{Proposition_1}
As $N\rightarrow \infty$, the upper bound $p^{(k)}_{up}$ of $p^{(k)}_{miss}$ becomes exact:
\begin{equation}
\lim_{N\uparrow \infty}\frac{p^{(k)}_{miss}}{p^{(k)}_{up}}=1,
\end{equation}
and $p^{(k)}_{miss}$ satisfies:
\begin{equation}\label{Eq_proposition_1}
\lim_{N\uparrow \infty}\frac{1}{N}\log p^{(k)}_{miss} = \lim_{N\uparrow \infty}\frac{1}{N}\log p^{(k)}_{up}= -I_1\left(\xi_{l^{(k)}_{opt}}^{(k)},\xi_{l^{(k)}_0}^{(k)}\right),
\end{equation}
where $l^{(k)}_0 \triangleq \arg\max_{l=1,\ldots,L^{(k)}, l\neq l^{(k)}_{opt}}\xi^{(k)}_l$.
\end{proposition}
\begin{proof}
This can be proved by using the results in Lemma~\ref{Lemma_LDP}, see Appendix~\ref{proof_P1}.
\end{proof}

Proposition~\ref{Proposition_1} demonstrates that the probability of misalignment at level $k$ decays exponentially with a rate $I_1\left(\xi_{l^{(k)}_{opt}}^{(k)},\xi_{l^{(k)}_0}^{(k)}\right)$ when $N$ is sufficiently large. To maximize $I_1\left(\xi_{l^{(k)}_{opt}}^{(k)},\xi_{l^{(k)}_0}^{(k)}\right)$ (thus achieve a faster decaying speed as $N$ increases), it is desirable to maximize the difference between $\xi^{(k)}_{l^{(k)}_{opt}}$ and $\xi^{(k)}_{l^{(k)}_0}$. Recalling that
$\xi^{(k)}_l = \frac{2P_Tg^{(k)}_{l}}{\sigma^2}$ with $g^{(k)}_{l} = |\alpha|^2 G^{(k)}_{l}(\psi,\phi)$, the difference between $\xi^{(k)}_{l^{(k)}_{opt}}$ and $\xi^{(k)}_{l^{(k)}_0}$ is increased by reducing $G^{(k)}_{l^{(k)}_0}(\psi,\phi)$ or increasing $G^{(k)}_{l^{(k)}_{opt}}(\psi,\phi)$. One way of reducing $G^{(k)}_{l^{(k)}_0}(\psi,\phi)$ is to impose the following constraint in synthesizing the transmit beamformer:
\begin{equation}\label{criteria_1}
W^{(k)}_{l^{(k)}}(\psi)=0, ~\text{if}~\psi\notin \Psi_{\mathbf{w}^{(k)}_{l}}.
\end{equation}
A similar constraint can be imposed for the synthesis of receive combiners.
As the AoD $\psi$ and the AoA $\phi$ vary both spatially and temporally, preference may not be given to any particular direction. In synthesizing the transmit beamformers, this imposes another constraint:
\begin{equation}\label{criteria_2}
W^{(k)}_{l^{(k)}}(\psi)=C, ~\text{if}~\psi\in \Psi_{\mathbf{w}^{(k)}_{l}},
\end{equation}
where $C$ is a constant to be maximized and is dependent on the beamwidth. Aggregating \eqref{criteria_1} and \eqref{criteria_2} yields the beam-synthesis criteria proposed in~\cite{Xiao2016,alkhateeb2014channel}.

The above analysis provides theoretical justifications for existing beam-synthesis criteria for mmWave beam-alignment. When the codebooks used for mmWave beam-alignment are properly designed, i.e., $\xi^{(k)}_{l^{(k)}_{opt}} \gg \xi^{(k)}_{l^{(k)}_0}\approx0$, the rate function $I_1\left(\xi_{l^{(k)}_{opt}}^{(k)},\xi_{l^{(k)}_0}^{(k)}\right)$ can be simplified and the upper bound is approximated as:
\begin{align}\label{LDP_appr}
p^{(k)}_{up} &\approx \left(L^{(k)}-1\right)\text{exp}{\left({-N\xi^{(k)}_{l^{(k)}_{opt}}/{4}}\right)} \\
&= \left(L^{(k)}-1\right)\text{exp}{\left(-NP_Tg^{(k)}_{l^{(k)}_{opt}}/{(2\sigma^2)}\right)}.
\end{align}

\par We validate our analysis for the ideal case with $\xi^{(k)}_l=0$, $\forall l\neq l^{(k)}_{opt}$ in Fig.~\ref{fig:single_level}. {Specifically, we consider $L^{(k)}=8$, $G^{(k)}=16$ and the signal to noise ratio without beamforming is $\text{SNR}=\frac{P_T|\alpha|^2}{\sigma^2}=-15$~dB. We plot the upper bound of the probability of misalignment derived in \eqref{upper_bound_f}, the approximation to the upper bound obtained from the LDP (given by \eqref{LDP_appr}) and the probability of misalignment obtained from simulations (according to \eqref{Eq:miss_align_1}) versus the total number of pilot symbols used for beam-alignment at level $k$, i.e., $NL^{(k)}$. It can be seen that the upper bound approaches  the simulation results as $N$ increases and the LDP approximation characterizes well the slope of the $p^{(k)}_{miss}$ as $N$ increases}. It can also be seen that the asymptotic results are consistent to simulations in a wide range of values of $N$ (e.g., when $NL^{(k)}$ is large such that $p^{(k)}_{miss}<10^{-2}$). This provides us confidence to rely on the asymptotic results to further analyze the performance of a $K$-level search.

\begin{figure}[t]
\centering
\includegraphics[width=0.5\textwidth]{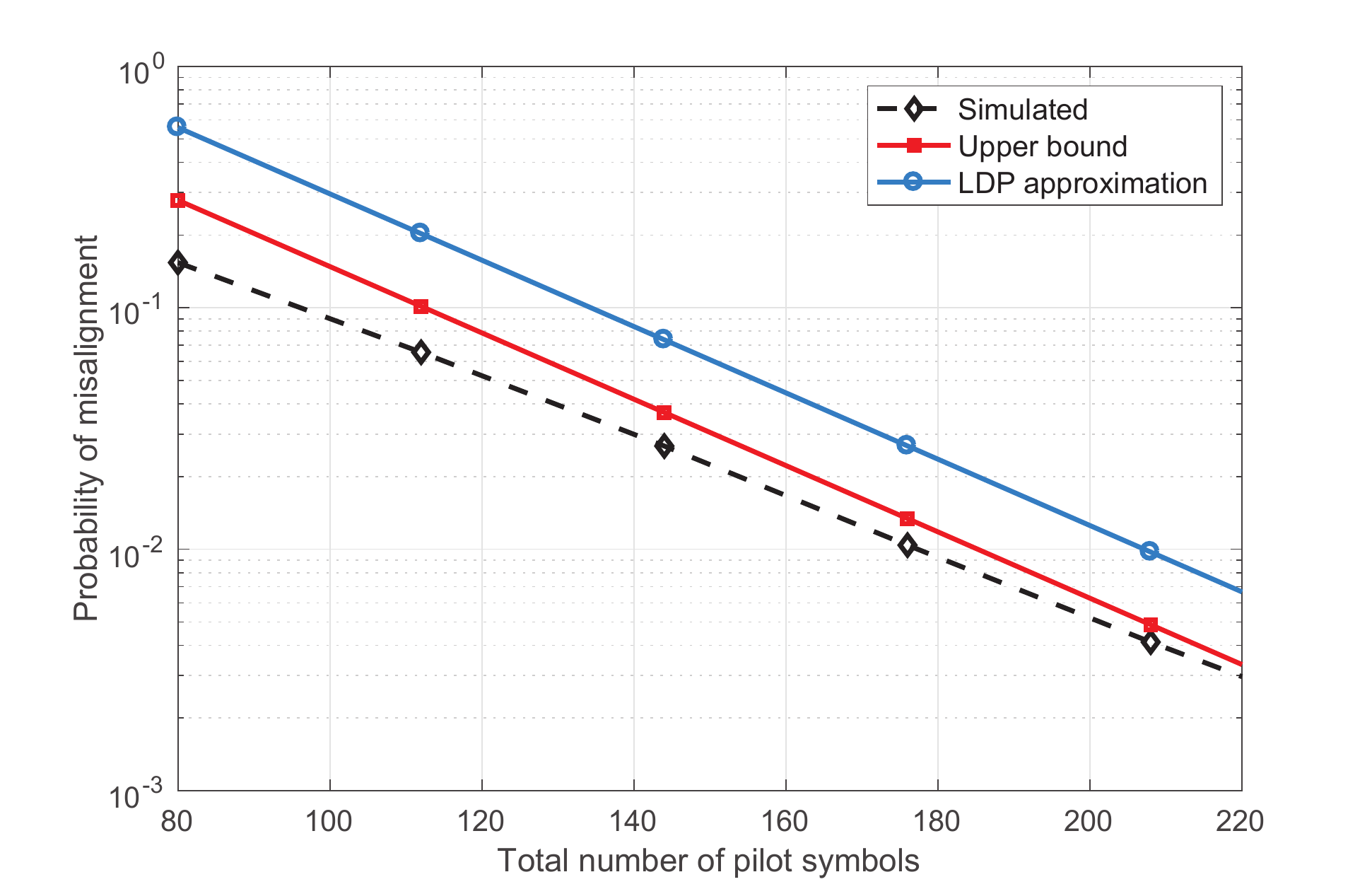}
\caption{Probability of misalignment versus the total number of pilot symbols at level $k$ (i.e, $NL^{(k)}$): $\text{SNR}=-15$ dB; $L^{(k)}=8$; $G^{(k)}=16$.}
\label{fig:single_level}
\end{figure}

\subsection{Overall Probability of Misalignment in a $K$-level Search}\label{sub_sec_multi}
As mentioned earlier in this section, the occurrence of ${\hat{l}^{(k)}_{opt}} \neq {l^{(k)}_{opt}}$ at any level leads to misalignment. The overall probability of misalignment during a $K$-level search can thus be represented as:
\begin{align}
p_{miss}(K) = \sum_{k=1}^{K}p^{(k)}_{miss}\prod_{m=0}^{k-1}\left[1-p^{(m)}_{miss}\right], \label{Overall_upperbound}
\end{align}
with a convention $p^{(0)}_{miss}=0$.

\par Using the results presented in Proposition~\ref{Proposition_1}, it can be proved that the asymptotic behavior of $p_{miss}(K)$ is determined by the level with the smallest decay rate function. We now present this result in the following proposition.

\begin{proposition}\label{Proposition_2}
 As $N\rightarrow \infty$, the overall probability of miss-detection $p_{miss}(K)$ is dominated by the miss-detection rate at the level with the smallest decay rate function. Namely,
 if $k^*=\arg\min_k I_1\left(\xi_{l^{(k)}_{opt}}^{(k)},\xi_{l^{(k)}_0}^{(k)}\right)$:
 \begin{equation}\label{Eq:proposition_2}
 \lim_{N\uparrow \infty} \frac{1}{N}\log{p_{miss}(K)} = -I_1\left(\xi_{l^{(k^*)}_{opt}}^{(k^*)},\xi_{l^{(k^*)}_0}^{(k^*)}\right).
 \end{equation}
\end{proposition}

\begin{proof}
See Appendix~\ref{proof_proposition_2}.
\end{proof}

\subsection{Hierarchical Search v.s. Exhaustive Search}\label{subsec:HSvsES}
Leveraging analytical results in Section~\ref{sub_sec_single} and Section~\ref{sub_sec_multi}, we now compare the relative performance between the hierarchical search and the exhaustive search in the asymptotic regime.
We confine our analysis by considering that the codebooks
have \emph{ideal} beam patterns that satisfy \eqref{criteria_1} and \eqref{criteria_2}. This means that each beamformer has a constant gain in its intended coverage interval and zero gain outside the interval (the coverage intervals of the codewords at the same level are not overlapping).
We further suppose that beamformers at the same level have the same beamwidth. In this case, the gain of the BS/UE beam within their coverage intervals is a constant and can be represented as:
\begin{equation}\label{beamforming_gain}
W^{(k)} = \frac{4\pi}{|\Omega_T|/L_T^{(k)}}~\text{and}~F^{(k)} = \frac{4\pi}{|\Omega_R|/L_R^{(k)}},
\end{equation}
where $\Omega_T$ and $\Omega_R$ are the solid angles spanned by the entire AoD and AoA interval, $\Phi$ and $\Psi$, respectively, and thus ${|\Omega_T|}/{L_T^{(k)}}$ and ${|\Omega_R|}/{L_R^{(k)}}$ correspond to the solid angles spanned by intervals $\Psi_{\mathbf{w}_{l}^{(k)}}$ and $\Phi_{\mathbf{f}_{l}^{(k)}}$, respectively.

\par It is noted that with these ideal beams, we have
 \begin{equation}\label{ideal_noncental_parameter}
 \xi^{(k)}_{l}=\left\{
    \begin{array}{ll}
         \frac{2|\alpha|^2 P_T W^{(k)} F^{(k)}}{\sigma^2}, & l=l^{(k)}_{opt},\\
        0,& \text{otherwise}.
    \end{array}
 \right.
 \end{equation}
From \eqref{beamforming_gain} and \eqref{ideal_noncental_parameter}, it can be seen that the $K$-level search is dominated by the first level, due to the widest beamwidth (thus the smallest beamforming gain) at this level.

We can now proceed to the performance comparison. For fair comparison, we consider that the total number of pilot symbols (thus the time) used for beam alignment is the same for both the hierarchical search and the exhaustive search, and is denoted as $N_{tot}$.

Recall that the total number of beamformer/combiner pairs examined in the $K$-level search is $L$ and the total number of beamformer/combiner pairs examined in the exhaustive search is $L_{ex}$ with $L_{ex} \gg L$. For each BS/UE beam pair, the numbers of pilot symbols transmitted in the $K$-level search and the exhaustive search are thus $N = N_{tot}/L$ and $N_{ex} = N_{tot}/L_{ex}$, respectively.

From Proposition~\ref{Proposition_1}, it can be seen that the probability of misalignment of the exhaustive search ($p^{ex}_{miss}$) satisfies
\begin{align}
\lim_{N_{tot}\uparrow \infty}\frac{1}{N_{tot}}\log p^{ex}_{miss} &= -\frac{1}{L_{ex}}\frac{\xi^{(K)}_{l^{(K)}_{opt}}}{4} \nonumber\\
&= -\frac{|\alpha|^2P_T}{ 2\sigma^2}\frac{4\pi}{|\Omega_T|}\frac{4\pi}{|\Omega_R|}, \label{eq:ex_compare1}\\
&\triangleq -I^{ex} \label{eq:ex_compare}
\end{align}
where \eqref{eq:ex_compare1} follows from~\eqref{beamforming_gain}-\eqref{ideal_noncental_parameter} and $L_{ex} = L^{(K)}_TL^{(K)}_R$. On the other hand, according to Proposition~\ref{Proposition_2}, it can be seen that the probability of misalignment of the $K$-level search ($p_{miss}(K)$) satisfies:
\begin{align}
\lim_{N_{tot}\uparrow \infty}\frac{1}{N_{tot}}\log p_{miss}(K) &= -\frac{1}{L}\frac{\xi^{(1)}_{l^{(1)}_{opt}}}{4} \nonumber\\
&=-\frac{L^{(1)}}{L}\frac{|\alpha|^2P_T}{ 2\sigma^2}\frac{4\pi}{|\Omega_T|}\frac{4\pi}{|\Omega_R|} \nonumber \\
&=-\frac{L^{(1)}}{L}I^{ex}, \label{eq:K_level_compare}
\end{align}
recalling that $L^{(1)} = L^{(1)}_TL^{(1)}_R$.

Since $L^{(1)}<L$ (as $L = \sum_{k=1}^{K}L^{(k)}$), comparing~\eqref{eq:ex_compare} and~\eqref{eq:K_level_compare}, it can be concluded that the exhaustive search yields a faster exponentially decaying rate than the $K$-level search when $N_{tot}$ is sufficiently large. This implies that for sufficiently large $N_{tot}$, the hierarchical search (with $K>1$) is outperformed by the exhaustive search. Although the analysis is carried out in the asymptotic regime, our numerical results demonstrate that this relative performance trend still holds for a wide range of finite $N_{tot}$, see Fig.~\ref{fig:single_path_p_miss} (note that the hierarchical search discussed so far corresponds to the curve with legend ``Hierarchical search (equal pilot length allocation)").

When producing Fig.~\ref{fig:single_path_p_miss}, we consider a ULA (with $N_T=64$) at the BS and a ULA (with $N_R=4$) at the UE, where {the AoD and AoA intervals simulated are assumed to be $[-30^{\circ},+30^{\circ}]$ and $[0^{\circ},360^{\circ}]$, respectively. Such a choice of AoD interval corresponds to the case where the BS intends to cover a $60^{\circ}$ sector.} At the BS, $K=5$ codebooks are employed with $L^{(k)}_T = 2^k$, $k=1,\ldots,5$, while at the UE, a fixed codebook with $4$ codewords are used at all the $K$ levels, with $L^{(k)}_R=4$, $k=1,\ldots,5$. At the first level of the hierarchical search, $2$ transmit beamformers and $4$ receive combiners are scanned. At levels $k>1$, the receive combiner is fixed at the best one detected at the first level, while $2$ finer-resolution transmit beamformers are scanned at each level. Therefore, $L^{(1)} = 2\times 4$ and $L^{(k)} = 2$ for $k \in [2,5]$, and $G^{(k)} = W^{(k)} F^{(k)}=2^{k+1}\times 4$ with ideal beam patterns.

From Fig.~\ref{fig:single_path_p_miss}, it can be seen that the exhaustive search outperforms the hierarchical search offering a lower probability of misalignment for any fixed $N_{tot}$. In other words, to reach a low targeted probability of misalignment, the required training overhead (in terms of total number of pilot symbols) by the exhaustive search is less than that of the hierarchical search. These results hold for a wide range of overhead levels, thus validating our asymptotic analysis.

\begin{figure}[t]
\centering
\includegraphics[width=0.5\textwidth]{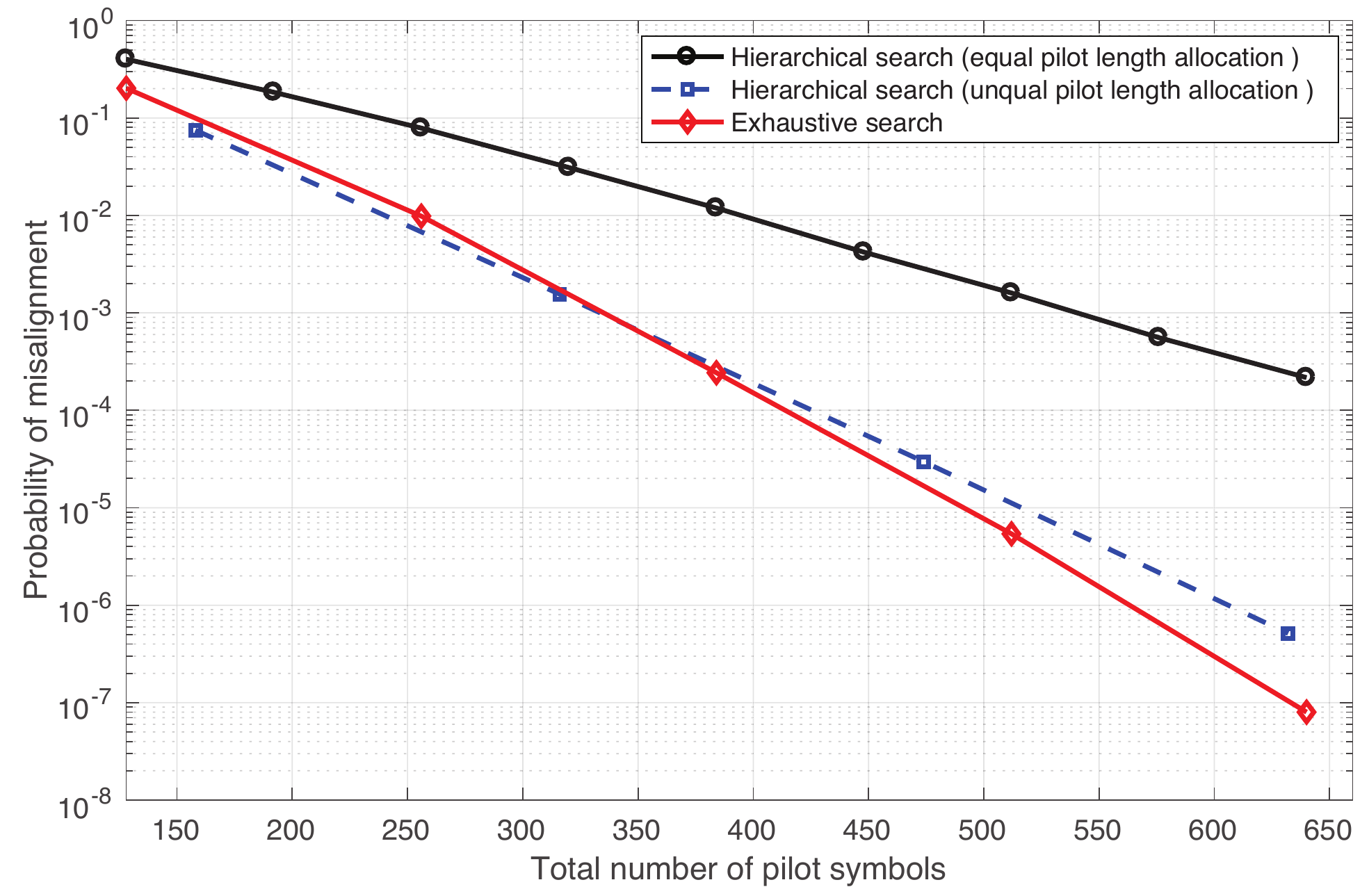}
\caption{Probability of misalignment versus the total number of pilot symbols ($N_{tot}$): $\text{SNR}=-15$ dB.}
\label{fig:single_path_p_miss}
\end{figure}

\begin{remark}\label{remark:unequal}
In the above comparison, the pilot length is made equal for each search level in a hierarchical search. Under such an assumption and with ideal beamformers, Proposition~\ref{Proposition_2} states that the first level will be the bottleneck on the overall misalignment performance for the hierarchical search. This can be seen from the rate functions that $\xi_{l^{(k)}_{opt}}^{(k)} \propto L_T^{(k)}L_R^{(k)}$ and that $\xi_{l^{(1)}_{opt}}^{(1)}$ is the minima among all $\xi_{l^{(k)}_{opt}}^{(k)}$s.

 \par To remove this bottleneck, one might consider allocating unequal pilot sequence length (denoted by $N^{(k)}$) to different search levels to equalize their performance. Towards this end, $N^{(k)}$ should be chosen to be inversely proportional to $L_T^{(k)}L_R^{(k)}$. {Without loss of optimality, assume that $N^{(k)} = N_{tot}\gamma /(L_T^{(k)}L_R^{(k)})$ for some $\gamma$ to be determined. Since $L^{(k)}$ beam pairs are examined for each level $k$, the total number of pilot symbols for a $K$-level search is given by $\sum_{k=1}^{K} L^{(k)}N^{(k)} = \sum_{k=1}^{K} L^{(k)}N_{tot}\gamma/(L_T^{(k)}L_R^{(k)}) =N_{tot}$. This hence implies that $\gamma = (\sum_{k=1}^{K} \frac{L^{(k)}}{L_T^{(k)}L_R^{(k)}})^{-1}$ (Note that $\gamma<1$ since $L^{(1)} = L_T^{(1)}L_R^{(1)}$).} As a result, the probability of misalignment at level $k$ satisfies $\frac{1}{N_{tot}}\lim_{N_{tot}\rightarrow \infty}p^{(k)}_{miss} =-\gamma I^{ex}$, following similar steps to those in~\eqref{eq:K_level_compare}. When $\gamma >L^{(1)}/L$, the hierarchical search with unequal pilot sequence length allocation will outperform the one with equal allocation as $N_{tot}$ goes large, see Fig.~\ref{fig:single_path_p_miss} for an illustration. {From this example, it can also be seen that the hierarchical search with unequal pilot sequence length allocation slightly outperform the exhaustive search when $N_{tot}$ is not that large. However, it is still dominated by the exhaustive search asymptotically, because $\gamma <1$.}
\par {Moreover, ideal beam patterns (as considered in the optimization of the pilot sequence lengths above) are not realizable in practice, for finite antenna arrays. With practical synthesized beamformers, when we allocated $N^{(k)}$'s as if beamformers were ideal,  no improvement was observed in our numerical experiments, see Section~\ref{sec:simulation}. At present, it is not clear that one can do better than the equal pilot length allocation in the hierarchical search.}\remarkend
\end{remark}

\section{Performance Evaluation Under Multi-Path Channel Model and Practical Beam Codebooks}\label{sec:simulation}
\par In the previous section, we have uncovered the relative performance between the exhaustive search and the hierarchical search under rank-one channel model and with ideal beam patterns. In this section, we generalize the study to the case with general multi-path channel model and practical synthesized beams. We will show by numerical results that the key insights obtained before continue to apply to these more realistic setups.

\subsection{Multi-Path Channel Model and Practical Beam Codebooks}
\par Following the measurement results from NYU~\cite{6834753}, the mmWave channels simulated are modeled with a limited number of multipath components, coming from distinct AoAs and AoDs. In particular, in the line-of-sight~(LOS) scenario, the channel is modeled as a Rician channel, in which a dominant path comes from AoA $\phi$ and AoD $\psi$. The Rician ${\cal K}$-factor (the ratio of the energy of the dominant path to the sum of the energy of the scattering components) is set to $13.2$ dB, according to the measurement results in \cite{muhi2010modelling}. In the non-LOS scenario, the channel is modeled as the sum of $M$ paths, each with a different AoA $\phi_m$ and AoD $\psi_m$, $m=1,\ldots,M$. The channel of each path is again assumed to be Rician, with ${\cal K}$-factor set to $6$ dB for all the paths~\cite{samimi201528}. The number of paths is $M=\max\{1,\varsigma\}$, where $\varsigma$ is a Poisson number with mean $1.8$, and the power fractions of the $M$ paths are generated by following the method in~\cite{6834753}. The SNR is defined as the ratio of the averaged total received power from all multipaths to the noise power.

\par Throughout the simulations, we consider the same AoA/AoD intervals, BS/UE antenna configuration and codebook setup (except with practical beamformers here) used to produce Fig.~\ref{fig:single_path_p_miss} (see Section~\ref{subsec:HSvsES}). Assume that the AoD $\psi$ ($\psi_m$s) is drawn uniformly from the intended AoD interval $[-30^{\circ},+30^{\circ}]$ (i.e., a $60^{\circ}$ sector), while the AoA $\phi$ ($\phi_m$s) of the multipath component is drawn uniformly from $[0^\circ,360^\circ]$. We adopt the {joint subarray and deactivation method by Xiao~\textit{et~al.}~\cite{Xiao2016} for constructing the analog beamforming codewords, see the relevant details in \cite[Section III.C-3)]{Xiao2016}}. The exhaustive search is carried out based on the transmit codebook at level $k=5$ and the receive codebook with 4 combiners. Thus the corresponding beam search space is $L_{ex} = 2^5 \times 4  = 128$. The results presented are calculated by averaging $10^6$ simulation trials.

\subsection{Results for LOS Scenario}
\par We first consider the LOS scenario. Fig.~\ref{Fig:p_miss_LOS} (a) and (b) plot the average misalignment probability of both the hierarchical search and the exhaustive search versus the total number of pilot symbols (maximum on the order of 1000 symbols) with $\text{SNR}= -15$ dB and $-10$ dB, respectively.

\begin{figure}[t]
\centering
\begin{subfigure}
\centering
\includegraphics[width=0.46\textwidth]{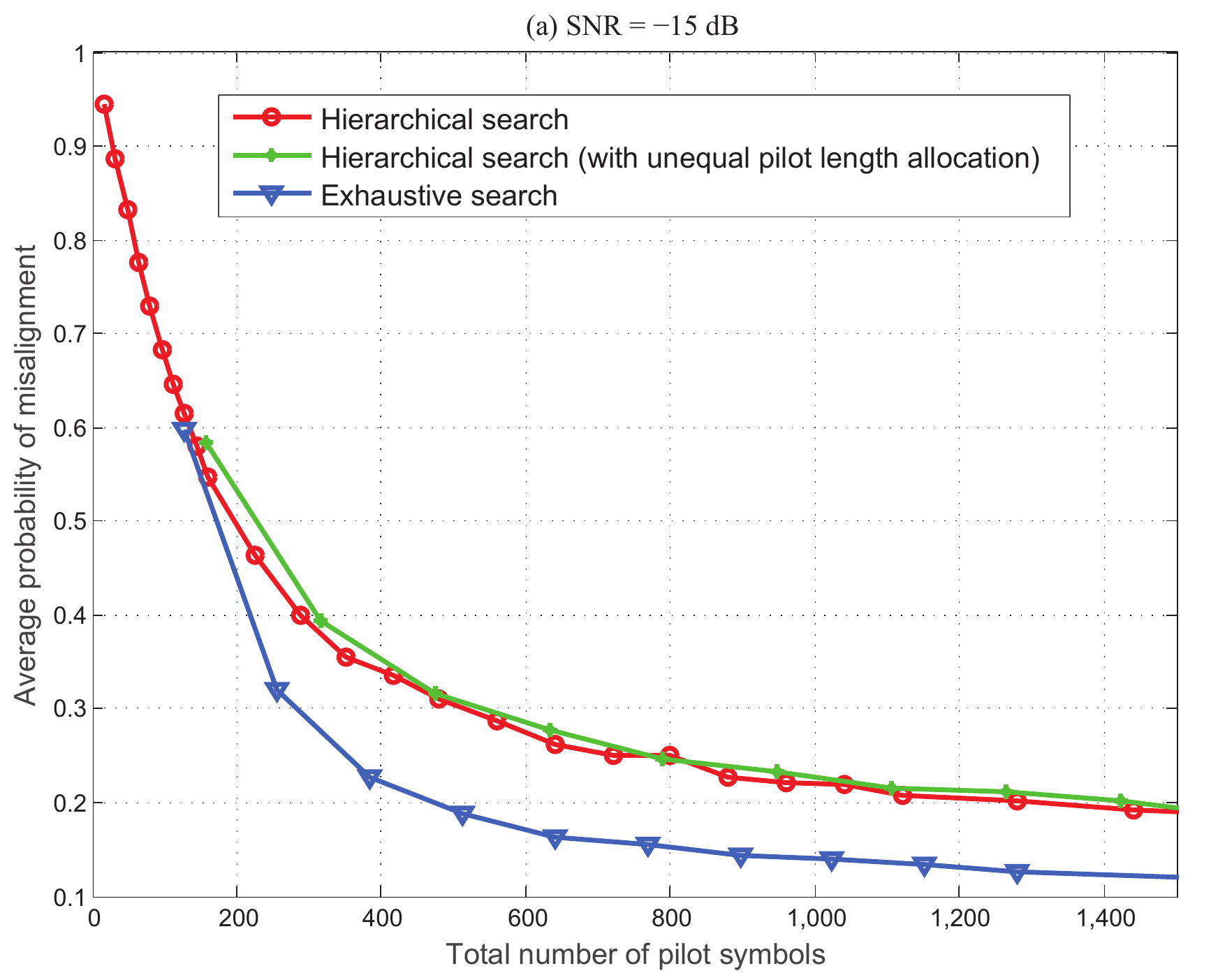}
\end{subfigure}
\begin{subfigure}
\centering
\includegraphics[width=0.46\textwidth]{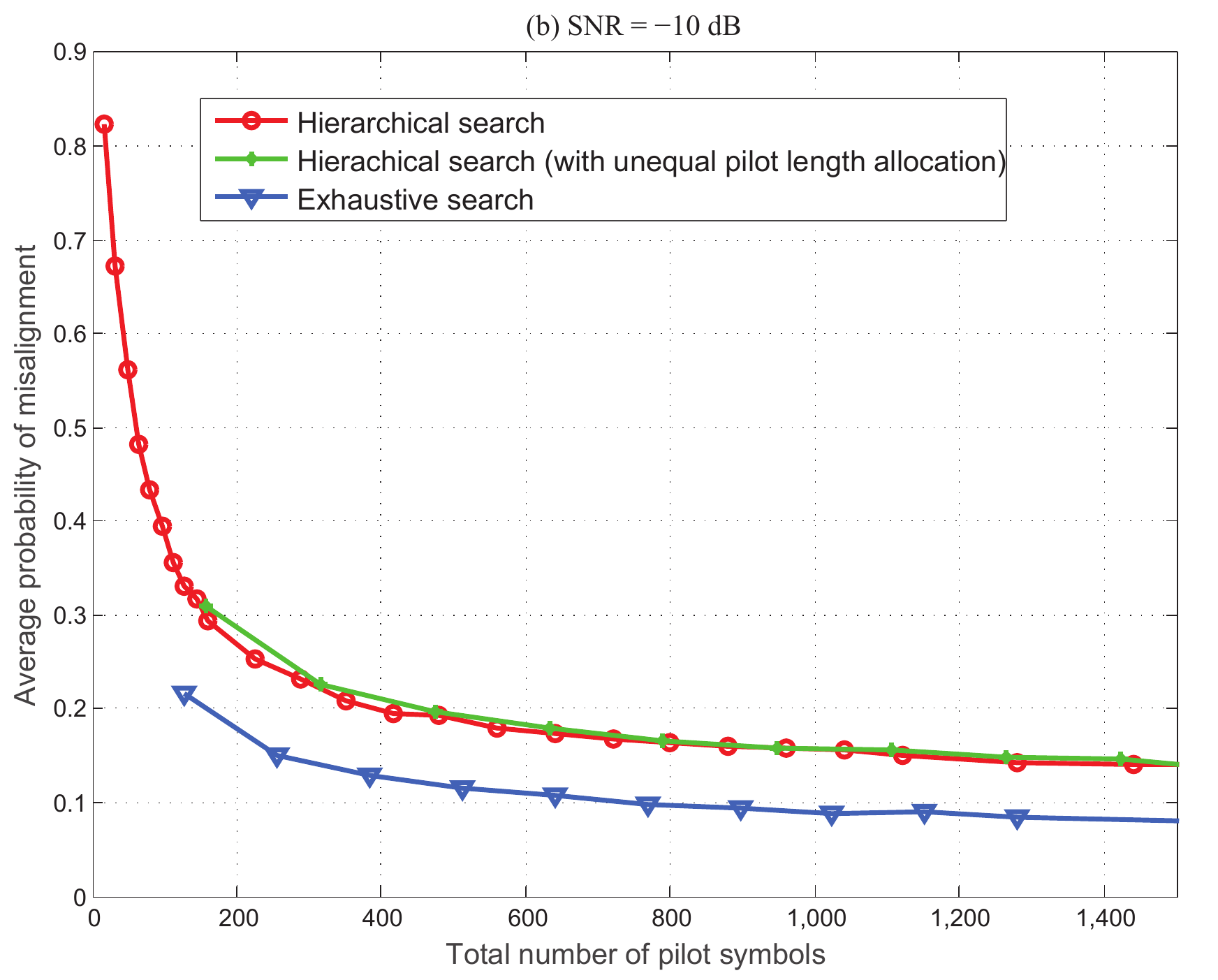}
\end{subfigure}
\caption{Average probability of misalignment: the LOS scenario.}
\label{Fig:p_miss_LOS}
\end{figure}


\begin{figure}[t]
\centering
\begin{subfigure}
\centering
\includegraphics[width=0.48\textwidth]{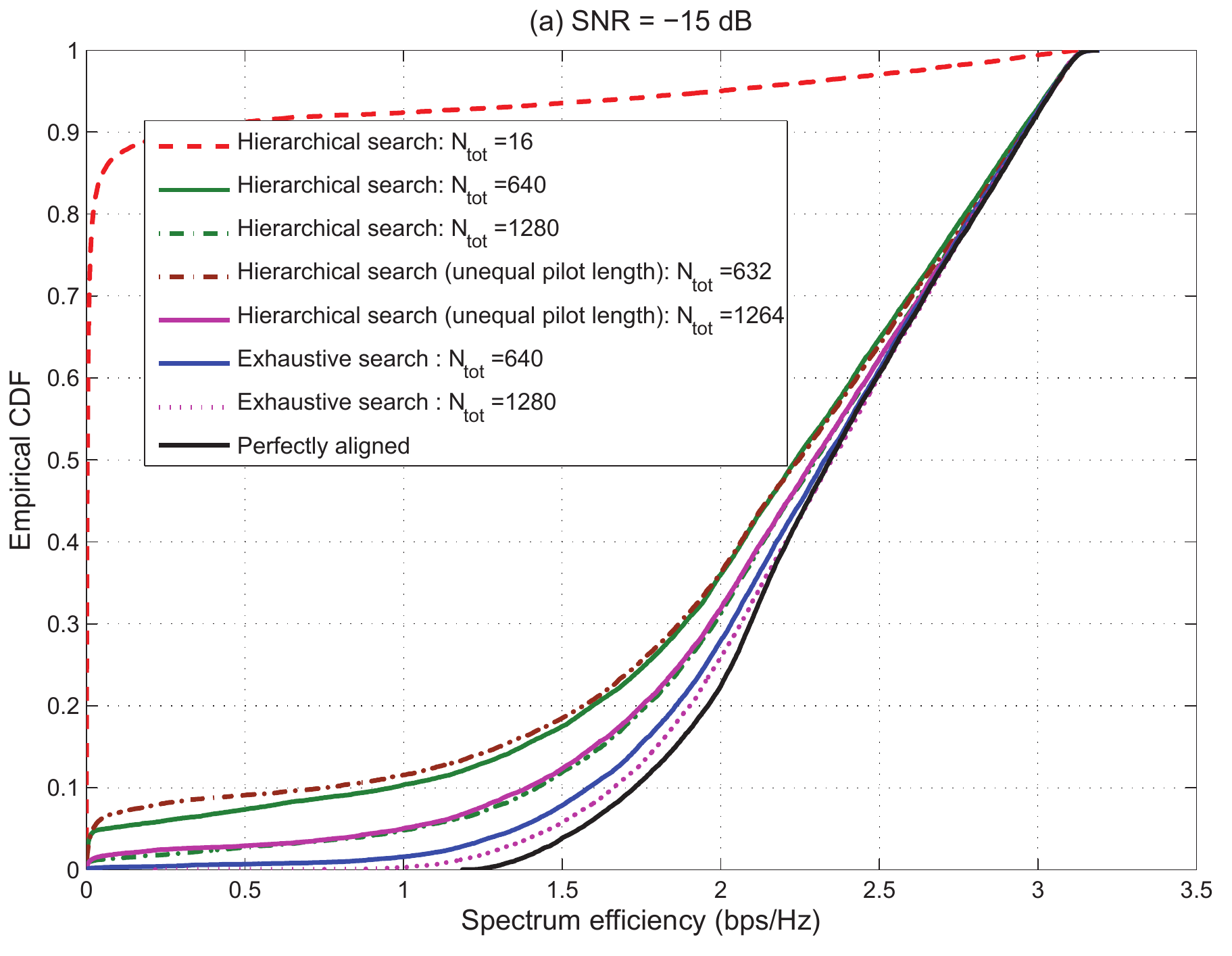}
\end{subfigure}
\begin{subfigure}
\centering
\includegraphics[width=0.48\textwidth]{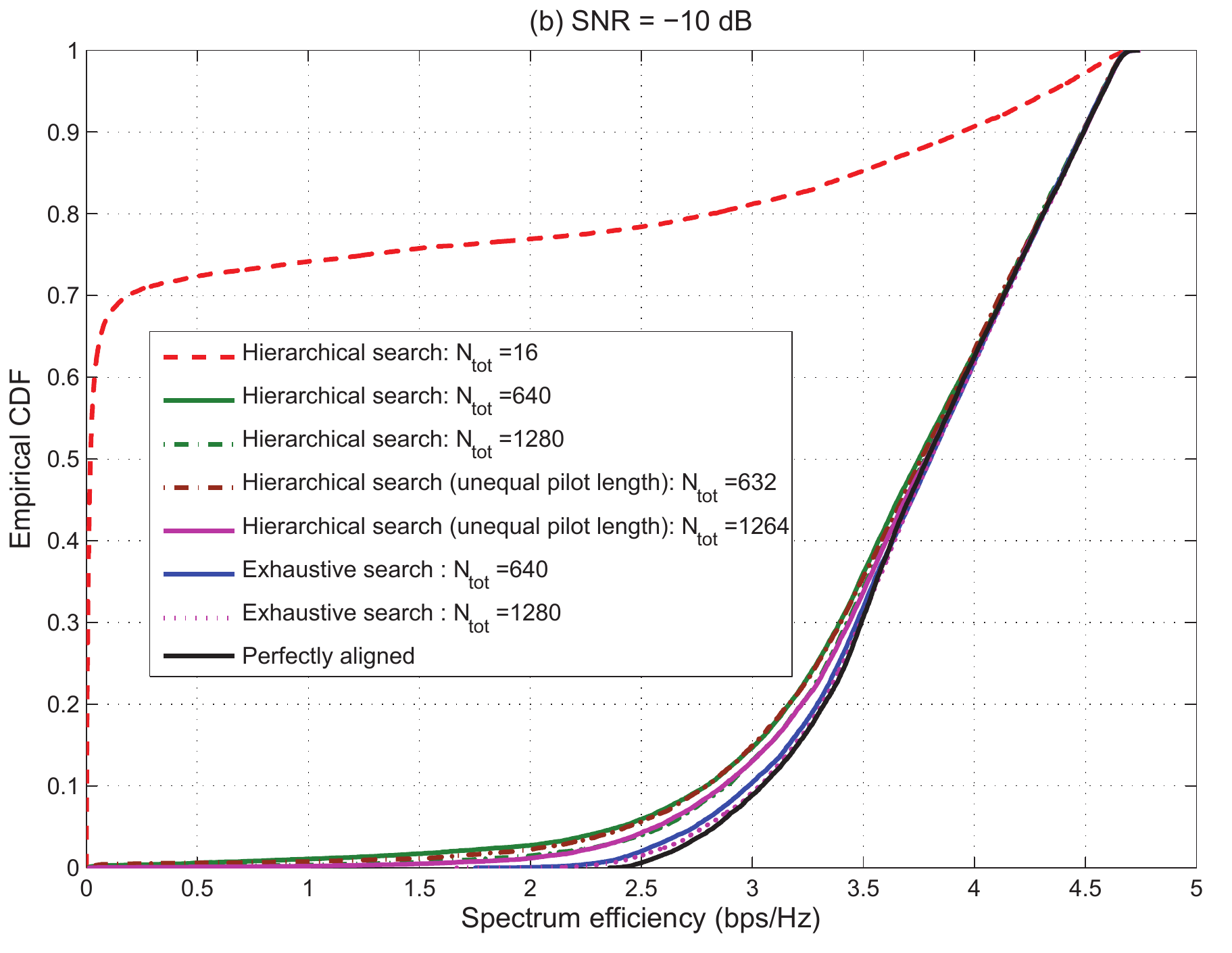}
\end{subfigure}
\caption{Empirical CDF of the spectrum efficiency after transmit beamforming and receive combining: the LOS scenario.}
\label{Fig:CDF_LOS}
\end{figure}

These choices of pilot symbol length and SNR values simulated are motivated by a practical design: Consider a mmWave communication system operating at $73$~GHz with a coherence bandwidth of $100$~MHz. Therefore, a training slot consisting of $1000$ pilot symbols amounts to a total training time of $10$ $\mu$s, which corresponds to only about $13.6\%$ of the coherence time for a UE of relatively low-mobility at $10$~m/s. In addition, suppose that mmWave BS transmits at peak power of $P_T = 15$~dBm, and the mmWave propagation follows an average path-loss model with relevant parameters (e.g., $\alpha = 69.8$ and $\beta = 2$) adopted from~\cite{6834753}. The SNRs considered thus correspond to typical (average) pre-beamforming SNR values for UEs located roughly $50$ meters and $30$ meters from the BS, respectively. It is certainly of importance to accommodate these UEs with fast and reliable beam search algorithms, as the mmWave cellular is expected to provide good outdoor coverage on the order of $100$ meters~\cite{6834753}.

\par It can be seen from Fig.~\ref{Fig:p_miss_LOS} that the exhaustive search outperforms the hierarchical search under the SNRs considered, as long as the number of pilot symbols is sufficiently large (e.g., $\ge$ 128 in this example) so that the exhaustive search is feasible. For instance, when $N_{tot} = 640$ and $\text{SNR}= -15$ dB, the exhaustive search provides $\frac{0.263-0.163}{0.263} = 38\%$ improvement over the hierarchical search in terms of misalignment probability. This performance trend generally agrees with the theoretical insights we have obtained under the rank-one (single-path) channel model.

\par We have also compared the achievable spectral efficiency under both the exhaustive search and the hierarchical search, {which is calculated according to the effective post-beamforming SNR using the best BS/UE beam pair determined in the beam training phase}. Fig.~\ref{Fig:CDF_LOS} plots the cumulative density function (CDF) of achievable spectral efficiency for both search strategies with the total number of pilot symbols $N_{tot}  = 16$, $640$ and $1280$, respectively. With $N_{tot} = 16$, the exhaustive search is infeasible in this case ({as $N_{tot}$ is smaller than the number of beam pairs to be evaluated, hence there is no result to plot}), and the hierarchical search achieves very limited spectral efficiency due to the poor misalignment performance. With more training resource, the spectral efficiency gets boosted under both the hierarchical search and the exhaustive search, as seen from Fig.~\ref{Fig:CDF_LOS}. However, a considerable performance difference is observed between these two search strategies, which is consistent with the misalignment probability results in Fig.~\ref{Fig:p_miss_LOS}.

\par In particular, the exhaustive search is found to achieve significantly better performance for the worst-case users. For instance, when $\text{SNR}=-15$ dB and $N_{tot} = 640$, the exhaustive search improves the $10$-percentile spectrum efficiency from $0.96$ bps/Hz (offered by the hierarchical search) to $1.59$ bps/Hz, an $65.6\%$ improvement. We have also observed that the performance difference diminishes in the high SNR regime (not shown here), since both the schemes can provide high reliability in beam-alignment.

\par In addition, we have included the performance of a hierarchical search with unequal pilot sequence allocation for each level as discussed in Remark~\ref{remark:unequal}. It can be confirmed that little benefit (in terms of both misalignment probability and achievable spectral efficiency) is reaped from this sophisticated allocation strategy in practical setups with non-ideal synthesized beamformers. Hence, we only focus on the strategy of equal pilot length allocation in the remaining results.

\begin{figure}[t]
\centering
\begin{subfigure}
\centering
\includegraphics[width=0.45\textwidth]{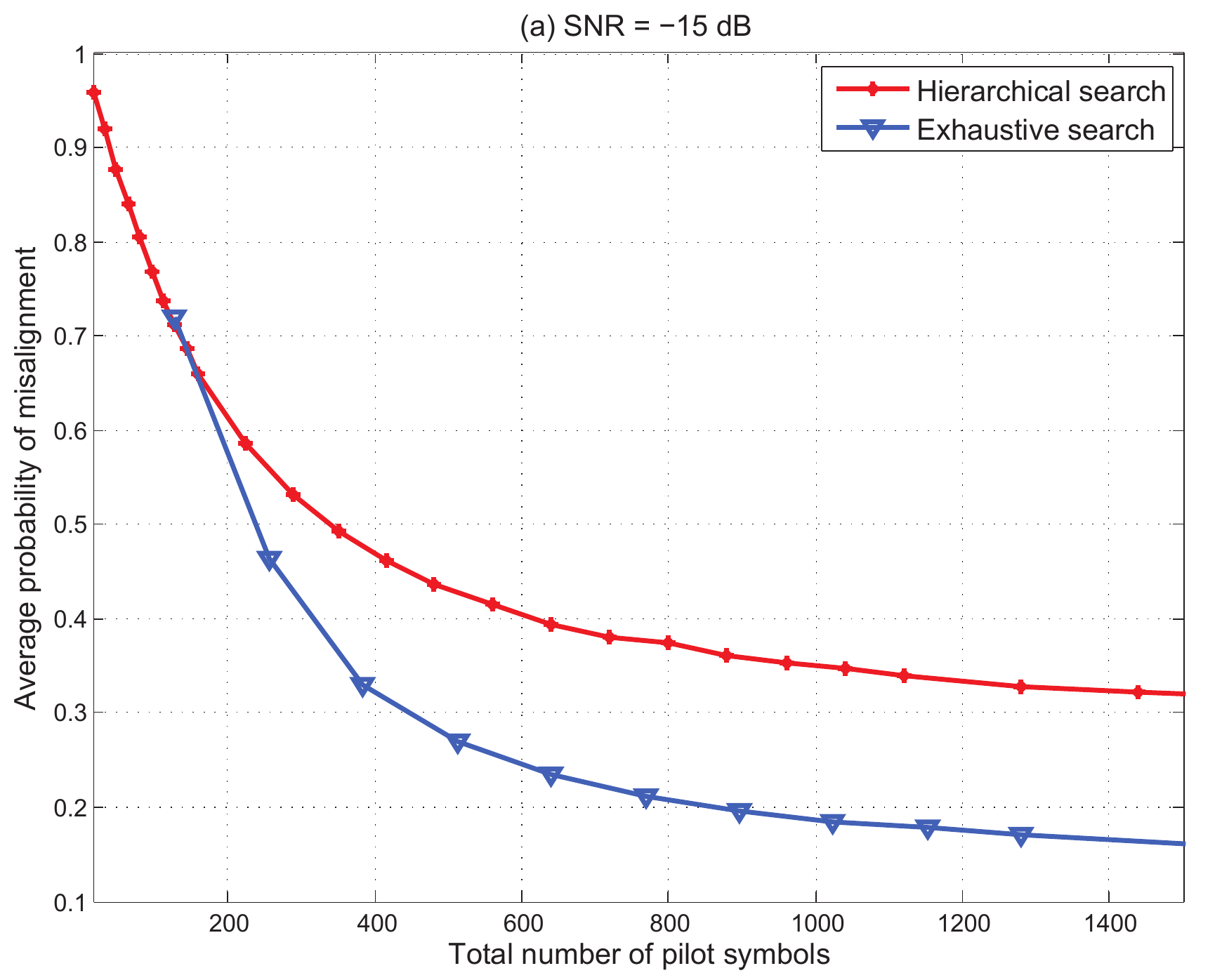}
\end{subfigure}
\begin{subfigure}
\centering
\includegraphics[width=0.45\textwidth]{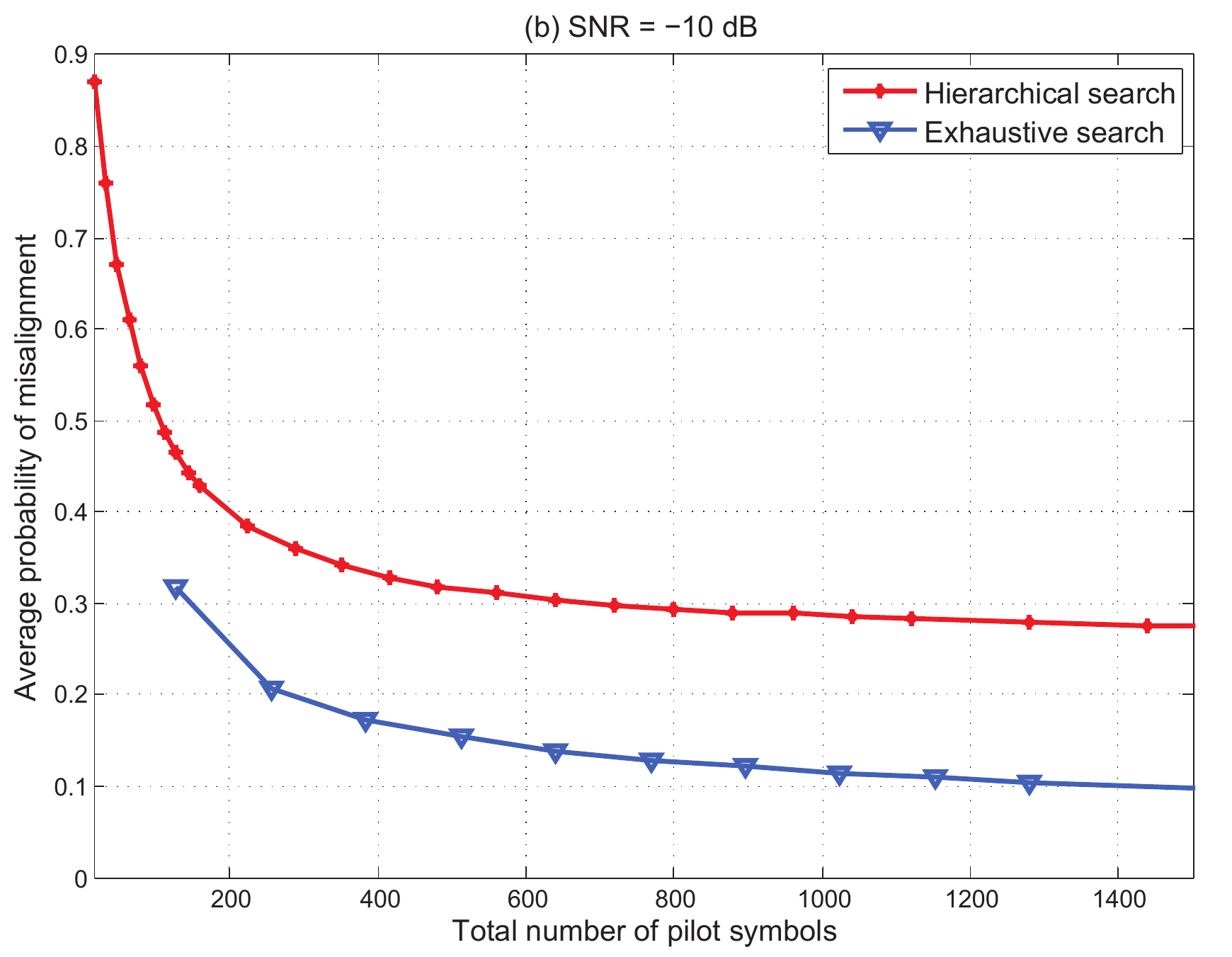}
\end{subfigure}
\caption{Average probability of misalignment: the non-LOS scenario.}
\label{Fig:p_miss_nonLOS}
\end{figure}

\subsection{Results for Non-LOS Scenario}

\begin{figure}[t]
\centering
\begin{subfigure}
\centering
\includegraphics[width=0.46\textwidth]{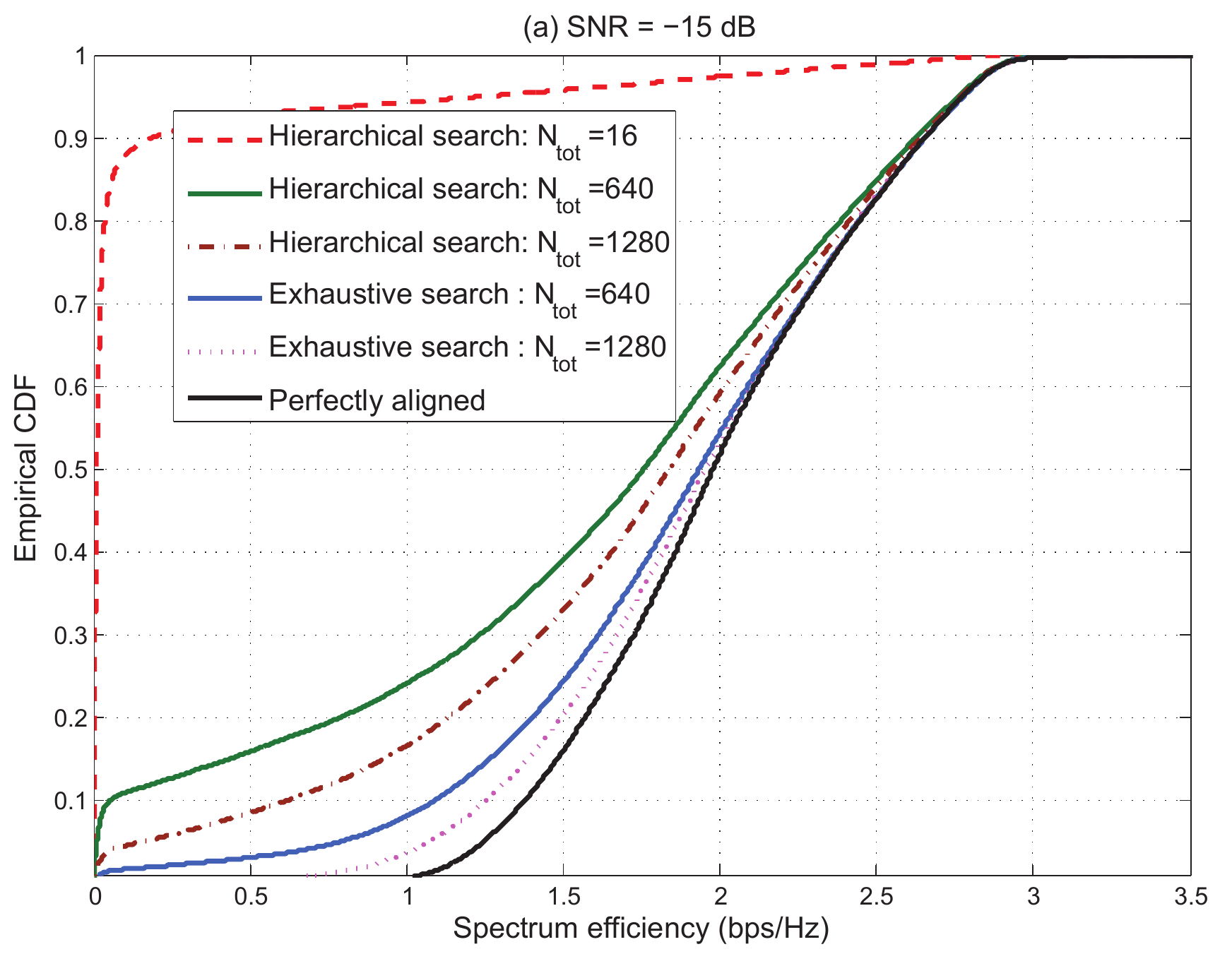}
\end{subfigure}
\begin{subfigure}
\centering
\includegraphics[width=0.46\textwidth]{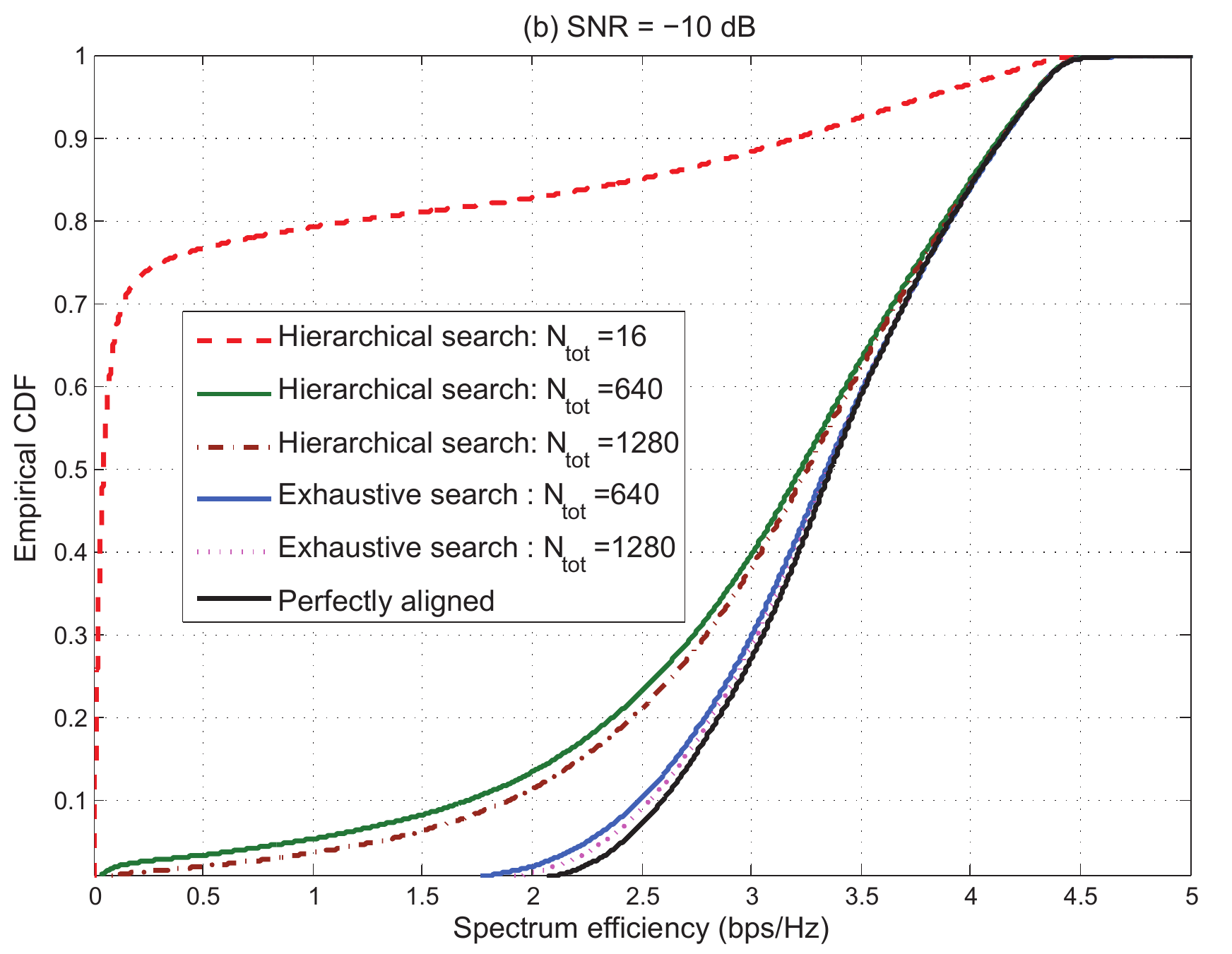}
\end{subfigure}
\caption{Empirical CDF of the spectrum efficiency after transmit beamforming and receive combining: the non-LOS scenario.}
\label{Fig:CDF_nonLOS}
\end{figure}

We now consider the NLOS scenario. Similar to the LOS evaluation, Fig.~\ref{Fig:p_miss_nonLOS} (a) and (b) plot the misalignment probability versus the number of pilot symbols under both the hierarchical search and the exhaustive search with $\text{SNR}= -15$ dB and $\text{SNR}= -10$ dB, respectively, while Fig.~\ref{Fig:CDF_nonLOS} (a) and (b) plot the CDF of the corresponding achievable spectral efficiency after beam-alignment.

The relative performance trends between the exhaustive search and the hierarchical search are similar to those in the LOS evaluation. But one may notice that the improvement of exhaustive search over hierarchical search (in terms of either misalignment probability or achievable spectral efficiency) in the NLOS scenario is more pronounced than that in the LOS scenario. For example, when $N_{tot} = 640$ and $\text{SNR}= -15$~dB, the exhaustive search now provides $\frac{0.394-0.235}{0.394} = 40.4\%$ improvement over the hierarchical search in terms of misalignment probability. This is due to the fact that signal received from the strongest multipath has a lower effective SNR than the reference SNR defined, since the received signal energy is split between the multipath components.

\section{Conclusions}
\par In this work, we have developed an analytical framework on the beam-alignment performance under both hierarchical search and exhaustive search for mmWave beam training, accounting for the training overhead in time. For a single-path channel model, we have derived bounds on the probability of misalignment and characterized the asymptotic behavior of these bounds using the method of large deviations. The exhaustive search is shown to be superior to the hierarchical search with sufficiently large training sequence length. The same performance trend still holds in the finite regime as shown by numerical results that have accounted for practically synthesized beams and multi-path fading channels, provided the pre-beamforming SNR is relatively low. This study indicates that while the hierarchical search can be an efficient design choice for users with high pre-beamforming SNR, the exhaustive search is a more effective method for a large portion of users that are not situated so close to mmWave base stations and tend to have relatively low SNR.

\appendices
\section{Proof of Lemma~\ref{Lemma_LDP}}\label{proof_LDP}
For notational convenience, we omit the superscript $^{(k)}$ for relevant variables if no confusion arises. Recall from~\eqref{Piece_wise_up} that the pairwise probability $F_{(2,2)}\left(1|\lambda_{l_{opt}},\lambda_{l}\right)$ can be represented as:
\begin{equation}
F_{(2,2)}\left(1|\lambda_{l_{opt}},\lambda_{l}\right) = Pr\left\{\frac{T_{l_{opt}}/N}{T_l/N}\leq 1\right\},
\end{equation}
where $T_{l_{opt}}\sim \chi^2_2(\lambda_{l_{opt}})$ and $T_l\sim \chi^2_2(\lambda_{l})$ are independent chi-square random variables with the same DoFs (DoFs=2) but different non-centrality parameters ($\lambda_{l_{opt}}$ and $\lambda_{l}$).

Denote $T'_{l_{opt}}\triangleq T_{l_{opt}}/N$ and $T'_l\triangleq T_l/N$. It is clear that $T'_{l_{opt}}$ and $T'_l$ are also independent. To prove \eqref{eq:lemma1_1}, we first show that $(T'_{l_{opt}},T'_l)$ satisfy LDP, using the well known Gartner-Ellis Theorem. To demonstrate this, we need to show that the limiting logarithmic moment generation function (MGF)
\begin{equation}\label{lim_MGF}
\Lambda(\mathbf{t}) = \lim_{N\uparrow \infty} \frac{1}{N}\Lambda_N(N\mathbf{t})
\end{equation}
exists as an extended real number~\cite{dembo2009large}, where $\Lambda_N(N\mathbf{t})\triangleq \log M_{(T'_{l_{opt}},T'_l)}(\mathbf{t})$ is the logarithmic MGF of $(T'_{l_{opt}},T'_l)$ and $\mathbf{t} = [t_1,t_2]$.

Since $T'_{l_{opt}}$ and $T'_l$ are independent, the MGF of the joint distribution is just the product of the MGF of each individual distribution. We therefore can represent $\Lambda_N(N\mathbf{t})$ as
\begin{align}\label{log_MGF}
\Lambda_N(N\mathbf{t}) &= \log M_{T'_{l_{opt}}}(Nt_1)+\log M_{T'_l}(Nt_2) \nonumber\\
& = \log \mathbb{E}_{T_{l_{opt}}}\left\{e^{t_1T_{l_{opt}}}\right\} + \log \mathbb{E}_{T_{l}}\left\{e^{t_2T_{l}}\right\},
\end{align}
where
\begin{equation} \label{MFG_U}
\mathbb{E}_{T_{l_{opt}}}\left\{e^{t_1T_{l_{opt}}}\right\}  =  \left\{
									\begin{array}{ll}
										\frac{e^{\frac{t_1\lambda_{l_{opt}} }{1-2t_1}}}{(1-2t_1)}, & t_1<\frac{1}{2} \\
										+\infty & \text{otherwise},
									   \end{array}
									   \right.
\end{equation}
\begin{equation} \label{MFG_V}
\mathbb{E}_{T_{l}}\left\{e^{t_2T_{l}}\right\}= \left\{
									\begin{array}{ll}
										\frac{e^{\frac{t_2\lambda_{l} }{1-2t_2}}}{(1-2t_2)}, & t_2<\frac{1}{2} \\
										+\infty & \text{otherwise}.
									   \end{array}
									   \right.
\end{equation}
Using \eqref{MFG_U}, \eqref{MFG_V}, \eqref{log_MGF} and \eqref{lim_MGF}, it can be shown that
\begin{align}
\Lambda(\mathbf{t}) &= \lim_{N\uparrow \infty} \frac{1}{N}\Lambda_{N}(N\mathbf{t}) \nonumber \\
& = \left\{
		\begin{array}{ll}
			\frac{t_1\xi_{l_{opt}}}{1-2t_1} + \frac{t_2\xi_{l}}{1-2t_2} & t_1,t_2<\frac{1}{2}\\
			+\infty. & \text{otherwise}
		\end{array}
		 \right.
\end{align}
Further, it can be shown that $\Lambda(\mathbf{t}) = 0$ only when $t_1<\frac{1}{2}$ and $t_2<\frac{1}{2}$. This verifies that $(T'_{l_{opt}},T'_l)$ satisfy the Gartner-Ellis conditions~\cite[Assumption 2.3.2, pp. 43]{dembo2009large}.

The rate function of $(T'_{l_{opt}},T'_l)$, i.e., $I(u,v)$ can then be obtained as
\begin{align}
I(u,v)& \triangleq \sup_{t_1,t_2 \in \mathbb{R}} \{t_1u+t_2v - \Lambda(\mathbf{t})\} \nonumber \\
& = \sup_{t_1,t_2<\frac{1}{2}} \left\{t_1u+t_2v - \frac{t_1\xi_{l_{opt}}}{1-2t_1} - \frac{t_2\xi_{l}}{1-2t_2}  \right\} \label{rate_function_optimization} \\
& = \frac{\left(\sqrt{\xi_{l_{opt}}}-\sqrt{u}\right)^2}{2}+\frac{\left(\sqrt{\xi_{l}}-\sqrt{v}\right)^2}{2}. \label{Rate_function}
\end{align}

With the rate function given in~\eqref{Rate_function}, the LDP tells us that
\begin{equation}
\lim_{N\uparrow \infty} \frac{1}{N}\log Pr\{(T'_{l_{opt}},T'_l) \in {\cal A}\} = -\inf_{(u,v) \in {\cal A}}I(u,v),
\end{equation}
if the set ${\cal A} \in \mathbb{R}^2$ is continuous.

For the pairwise probability $F_{(2,2)}\left(1|\lambda_{l_{opt}},\lambda_{l}\right)$, we have ${\cal A}=\{(u,v)|u \leq v\}$. This leads to the following rate function:
\begin{equation} \label{optimisation_rate}
I_1\left(\xi_{l_{opt}},\xi_{l}\right)= \inf_{u/v \leq 1}I(u,v).
\end{equation}

If the infimum of $I(u,v)$ in ${u/v \leq 1}$ is attained when $u/v<1$, then using the Karush-Kuhn-Tucker (KKT) conditions~\cite{boyd2004convex} of \eqref{optimisation_rate}, it can be shown that the infimum is attained at $u=\xi_{l_{opt}}$ and $v=\xi_l$. However, this contradicts with the pre-assumption that $u/v<1$ since  $\xi_{l_{opt}}>\xi_l$. It can therefore be concluded that the infimum of $I(u,v)$ in ${u/v \leq 1}$ is attained when $u=v$:
\begin{equation} \label{optimisation_rate_2}
I_1\left(\xi_{l_{opt}},\xi_{l}\right)= \inf_{u/v \leq 1}I(u,v)=\inf_{u/v = 1}I(u,v).
\end{equation}
It is then straightforward to show that the infimum is achieved when $\sqrt{u}=\sqrt{v} = \frac{\sqrt{\xi_{l_{opt}}}+\sqrt{\xi_{l}}}{2}$, and $I_1\left(\xi_{l_{opt}},\xi_{l}\right)=\frac{\left(\sqrt{\xi_{l_{opt}}}-\sqrt{\xi_{l}}\right)^2}{4}$. This proves \eqref{eq:lemma1_1}.

\par One can prove Eq.\eqref{eq:lemma1_2} in a similar manner and thus the details are omitted here.

\section{Proof of Proposition~\ref{Proposition_1}}\label{proof_P1}
\par To prove~\eqref{Eq_proposition_1}, we need to prove that $\lim_{N\rightarrow \infty}\frac{p^{(k)}_{up}}{p^{(k)}_{low}}=1$.
\par Recall that
\begin{equation}
p^{(k)}_{up} = \sum_{l=1,l\neq l^{(k)}_{opt}}^{L^{(k)}} F_{(2,2)}\left(1|\lambda_{l^{(k)}_{opt}}^{(k)},\lambda_{l}^{(k)}\right).
\end{equation}
Let $l^{(k)}_0 = \arg\max_{l=1,\ldots,L^{(k)},l\neq l^{(k)}_{opt}}\xi^{(k)}_l$. As $\xi^{(k)}_l<\xi^{(k)}_{l^{(k)}_{opt}}$, $\forall~l\neq l^{(k)}_{opt}$, it can be seen from \eqref{rate_1} that
\begin{equation}\label{Eq:dif_decay}
I_1\left(\xi_{l^{(k)}_{opt}}^{(k)},\xi_{l_0}^{(k)}\right)<I_1\left(\xi_{l^{(k)}_{opt}}^{(k)},\xi_{l}^{(k)}\right),~\forall l\neq l^{(k)}_0, l\neq l^{(k)}_{opt}.
\end{equation}
This implies that the $(L^{(k)}-1)$ terms in $p^{(k)}_{up}$ decay exponentially at different rates, with the term indexed by $l^{(k)}_0$ decaying at the lowest rate. Mathematically, it means:
\begin{align}
&\lim_{N\uparrow \infty} \frac{1}{N}\log\frac{F_{(2,2)}\left(1|\lambda_{l^{(k)}_{opt}}^{(k)},\lambda_{l}^{(k)}\right)}{F_{(2,2)}\left(1|\lambda_{l^{(k)}_{opt}}^{(k)},\lambda_{l_0}^{(k)}\right)} \nonumber \\ &~~~~~~~=I_1\left(\xi_{l^{(k)}_{opt}}^{(k)},\xi_{l_0}^{(k)}\right)-I_1\left(\xi_{l^{(k)}_{opt}}^{(k)},\xi_{l}^{(k)}\right)>0,
\end{align}
or equivalently,
\begin{equation}\label{asymp_dominant}
\lim_{N\uparrow \infty} \frac{F_{(2,2)}\left(1|\lambda_{l^{(k)}_{opt}}^{(k)},\lambda_{l}^{(k)}\right)}{F_{(2,2)}\left(1|\lambda_{l^{(k)}_{opt}}^{(k)},\lambda_{l_0}^{(k)}\right)} = 0.
\end{equation}
 Eq.~\eqref{asymp_dominant} demonstrates that $p^{(k)}_{up}$ is dominated by the term, $F_{(2,2)}\left(1|\lambda_{l^{(k)}_{opt}}^{(k)},\lambda_{l_0}^{(k)}\right)$, when $N\rightarrow \infty$:
\begin{equation}
\lim_{N\uparrow \infty} \frac{p^{(k)}_{up}}{F_{(2,2)}\left(1|\lambda_{l^{(k)}_{opt}}^{(k)},\lambda_{l_0}^{(k)}\right)} = 1.
\end{equation}
Further, since
\begin{align}
I_2\left(\xi_{l^{(k)}_{opt}}^{(k)},\xi_{i}^{(k)},\xi_{j}^{(k)}\right) &= \frac{\left(\sqrt{2\xi^{(k)}_{l^{(k)}_{opt}}}-\sqrt{\xi^{(k)}_{i}+\xi^{(k)}_{j}}\right)^2}{6}\nonumber \\
& = \frac{\left(\sqrt{\xi^{(k)}_{l^{(k)}_{opt}}}-\sqrt{\xi^{(k)}_{i,j}}\right)^2}{3},
\end{align}
where $\xi^{(k)}_{i,j}=\frac{\xi^{(k)}_{j}+\xi^{(k)}_{j}}{2}<\max\left\{\xi^{(k)}_{j},\xi^{(k)}_{j}\right\}\leq \xi_{l_0}^{(k)}<\xi^{(k)}_{l^{(k)}_{opt}}$, it can be seen that $I_2\left(\xi_{l^{(k)}_{opt}}^{(k)},\xi_{i}^{(k)},\xi_{j}^{(k)}\right)>I_1\left(\xi_{l^{(k)}_{opt}}^{(k)},\xi_{l_0}^{(k)}\right)$. This implies that
\begin{equation}
\lim_{N\uparrow \infty}\frac{F_{(2,4)}\left(1|\lambda_{l^{(k)}_{opt}}^{(k)},\lambda_{i}^{(k)}+\lambda_{j}^{(k)}\right)}{F_{(2,2)}\left(1|\lambda_{l^{(k)}_{opt}}^{(k)},\lambda_{l_0}^{(k)}\right)}=0,
\end{equation}
and hence,
\begin{align}
\lim_{N\uparrow \infty} \frac{p^{(k)}_{low}}{F_{(2,2)}\left(1|\lambda_{l^{(k)}_{opt}}^{(k)},\lambda_{l_0}^{(k)}\right)}= \lim_{N\uparrow \infty} \frac{p^{(k)}_{up}}{F_{(2,2)}\left(1|\lambda_{l^{(k)}_{opt}}^{(k)},\lambda_{l_0}^{(k)}\right)} = 1.
\end{align}
It is then immediate to see that $\lim_{N\rightarrow \infty}\frac{p^{(k)}_{up}}{p^{(k)}_{low}}=1$, and thus~\eqref{Eq_proposition_1} follows. This completes the proof.

\section{Proof of Proposition~\ref{Proposition_2}}\label{proof_proposition_2}
Let $k^*=\arg\min_k I_1\left(\xi_{l^{(k)}_{opt}}^{(k)},\xi_{l_0}^{(k)}\right)$. For any $j \neq k^*$, using \eqref{Eq_proposition_1} in Proposition~\ref{Proposition_1}, it can be seen that
\begin{align}
&\lim_{N\uparrow \infty}\frac{1}{N}\log\frac{p^{(j)}_{miss}\prod_{m=0}^{j-1}\left[1-p^{(m)}_{miss})\right]}{p^{(k^*)}_{miss}}\nonumber \\
&= I_1\left(\xi_{l^{(k^*)}_{opt}}^{(k^*)},\xi_{l^{(k^*)}_0}^{(k^*)}\right)-I_1\left(\xi_{l^{(j)}_{opt}}^{(j)},\xi_{l^{(j)}_0}^{(j)}\right)<0,
\end{align}
which hence implies that
\begin{align}
\lim_{N\uparrow \infty}\frac{p^{(j)}_{miss}\prod_{m=0}^{j-1}\left[1-p^{(m)}_{miss})\right]}{p^{(k^*)}_{miss}}=0.
\end{align}
Recalling~\eqref{Overall_upperbound}, it is then clear that
\begin{align}
\lim_{N\uparrow \infty}\frac{p_{miss}(K)}{p^{(k^*)}_{miss}} = 1
\end{align}
and hence
\begin{align}
\lim_{N\uparrow \infty} \frac{1}{N}\log{p_{miss}(K)} &= \lim_{N\uparrow \infty} \frac{1}{N}\log{p^{(k^*)}_{miss}}\nonumber \\
&= -I_1\left(\xi_{l^{(k^*)}_{opt}}^{(k^*)},\xi_{l^{(k^*)}_0}^{(k^*)}\right),
\end{align}
which completes the proof.

\bibliographystyle{IEEETran}


\end{document}